\documentclass{article}

\RequirePackage{amsthm,amsmath,amsfonts,amssymb}
\RequirePackage[authoryear]{natbib}
\usepackage{hyperref}
\RequirePackage{graphicx}
\usepackage{float}
\usepackage{algorithm,algpseudocode}
\usepackage{booktabs}
\usepackage{empheq}
\usepackage{cancel}
\usepackage[utf8]{inputenc}
\usepackage{listings}
\usepackage{mathrsfs}
\usepackage{array}
\usepackage{authblk}
\usepackage{orcidlink}

\usepackage{todonotes}
\usepackage{subfigure}
\usepackage{adjustbox}
\usepackage{multirow}
\usepackage{amsmath}
\usepackage{bbm}

\newcommand{\bs}{\boldsymbol}

\newcommand{\x}{\mathbf{x}}

\newcommand{\K}{\mathbf{K}}

\newcommand{\C}{\mathbf{C}}

\newcommand{\Q}{\mathbf{Q}}

\newcommand{\U}{\mathbf{U}}

\newcommand{\X}{\mathbf{X}}

\newcommand{\T}{{\!\mathsf{T}}}  %

\DeclareMathOperator{\logit}{logit}
\newcommand{\dif}{\ensuremath{\mathrm{d}}}

\theoremstyle{plain}

\newtheorem{theorem}{Theorem}[section]

\newtheorem{proposition}[theorem]{Proposition}

\theoremstyle{definition}
\newtheorem{definition}[theorem]{Definition}
\newtheorem{remark}[theorem]{Remark}
\newtheorem*{example}{Example}

\title{Tail Halos: The Covariate Dual of the Tail}
\author[1]{Gianmarco Callegher\,\orcidlink{0009-0001-1020-7887}}
\author[2]{Migue de Carvalho\,\orcidlink{0000-0003-3248-6984}}

\affil[1]{University of G\"ottingen}
\affil[1]{University of Edinburgh}
\date{\today}

\begin{document}

\maketitle

\begin{abstract}
Extreme events are often studied through regression models describing how covariates modify tail parameters. In many applications, however, interest lies directly in identifying the parts of the covariate space associated with tail events and with elevated exceedance probabilities. We formalize these ideas through \emph{tail halos}, covariate-defined regions associated with extremes. Risk halos, when they exist, are minimum-mass covariate regions accounting for a prescribed fraction of exceedance probability, while favorable halos collect covariate values where conditional exceedance probability exceeds its marginal level. To make these set-valued objects interpretable beyond low dimensions, we introduce halo-induced covariate laws and, when densities exist, corresponding \emph{halo densities}. We develop a Bayesian structured additive distributional regression framework based on extended generalized Pareto marginals, covariate-dependent copulas, and spike-and-slab effect selection to learn marginal and joint halo-induced covariate laws while propagating posterior uncertainty. The approach extends Bayesian shrinkage methods for conditional tail modeling to multivariate bulk-and-tail modeling, nonlinear effect selection, and covariate-region inference. Simulations show recovery of relevant nonlinear effects and increasingly accurate finite-sample halo representations as sample size grows. An application to PM$_{2.5}$ and NO$_2$ extremes in Edinburgh reveals distinct environmental and temporal configurations associated with high pollution.
\end{abstract}

\section{Introduction}\label{introduction}
Extreme events are rarely extreme in isolation. Heatwaves depend on atmospheric conditions, air-pollution episodes on weather and human activity, and financial crashes on market conditions that evolve over time. A central goal of extreme-value statistics is therefore not only to estimate how rare such events are, but also to understand the covariate configurations under which they occur. This is difficult because the relevant information lies in the tail, where data are scarce and extrapolation is unavoidable. The field of statistics of extremes provides the natural framework for this task, with methods designed to model the probability of events beyond the range of the observed data \citep{beirlant2004, coles2001, HandbookExtremes2026, dehaan2006}.

Regression models for extremes usually approach the problem by allowing the parameters of an extreme-value distribution---for example, a generalized extreme value, generalized Pareto, or extended generalized Pareto distribution---to depend on covariates \citep{chavez2005, davison1990, decarvalho2021a, decarvalho2022, decarvalho2026regression, wang2009}. In multivariate settings, related ideas have been used to let extremal dependence vary with covariates \citep{castro2018, mhalla2019b, lee2024}. This parameter-first viewpoint is powerful and interpretable, but it leaves an important question only indirectly answered. An analyst may know that a tail index, scale parameter, or dependence coefficient changes with temperature, wind, time of day, or season, while still not knowing which covariate configurations carry most exceedance probability, or where exceedances are more likely than they are on average.

Variable selection is especially important in this setting, because covariates may affect different parts of the conditional distribution in different ways. The extreme-value Bayesian Lasso of \citet{decarvalho2021a} provides a Bayesian variable-selection framework for learning covariate effects on the conditional left and right tails of a univariate response. Our approach builds on this idea, but replaces the Lasso-type shrinkage with spike-and-slab selection for nonlinear structured additive effects, extends the construction to multivariate responses through covariate-dependent copulas, and uses the resulting fitted tail probabilities to define covariate regions associated with extremes.

In this paper, we make this covariate-side question the inferential target. We introduce \emph{tail halos}, set-valued summaries of the covariate space associated with extreme responses. We study two complementary versions. A \emph{risk halo} is a covariate region of minimum probability mass among those accounting for at least an $\alpha$-fraction of all responses exceeding a threshold $u$. Under an atom-free condition it exists, is unique up to null sets, and is a superlevel set of the conditional exceedance probability. The \emph{favorable halo}, motivated by the notion of favorable events \citep[e.g.,][]{kramer2005}, is the region where the conditional exceedance probability exceeds its marginal level. Risk halos therefore identify where tail mass is concentrated, while favorable halos identify where tail risk is elevated relative to the background.

A further contribution is to attach to each halo a halo-induced covariate law and, when densities exist, the corresponding \emph{halo densities}. This is a central part of the construction. The halo itself is a subset of covariate space, but such sets quickly become difficult to display, compare, or communicate when the covariate dimension grows. Halo densities turn the set-valued object into ordinary distributional summaries of the covariates inside the halo, while preserving the population target. This provides the theoretical bridge between the definitions of the halos and the empirical summaries used later for estimation and uncertainty quantification.

We then develop a Bayesian structured additive distributional regression framework for learning halo-induced covariate laws from data. The marginal model is based on the extended generalized Pareto distribution \citep{naveau2016, papastathopoulos2013}, thereby including the extended generalized Pareto regression model of \citet{decarvalho2021a} as a special case. A key part of the construction is a spike-and-slab prior on nonlinear P-spline effects, which performs effect selection separately for each marginal and copula parameter. This lets the model distinguish covariates that affect marginal tail behavior from those that affect extremal dependence. The computational strategy is connected to recent work on stochastic variational inference for structured additive distributional regression \citep{callegher2025stochasticvariationalinferencestructured}, but the inferential target here is different: we propagate posterior uncertainty from fitted marginal and joint tail probabilities to the covariate laws induced by the corresponding halos. In the multivariate case, the model combines conditional EGPD marginals with a covariate-dependent copula, allowing covariates to affect marginal tails without necessarily affecting joint tail dependence, and vice versa.

The remainder of the paper is organized as follows. Section~\ref{model}
introduces risk and favorable halos, together with the halo-induced covariate laws used to
visualize them. Section~\ref{model2} develops the Bayesian structured additive
distributional regression model used to learn these objects from data.
Section~\ref{sec:simulation} reports the main findings of a simulation study. A
data application is presented in Section~\ref{sec:application}. We conclude in
Section~\ref{sec:discussion} with a discussion and directions for future work.

\section{Tail halos: theory and examples}\label{model}
\subsection{Risk halo}
We first introduce the notion of a risk halo in the univariate-response setting. Let $Y$ take values in $\mathcal{Y} \subseteq \mathbb{R}$ and let $\mathbf{X} \sim P_{\mathbf{X}}$ take values in a Borel set $\mathcal{X} \subseteq \mathbb{R}^p$. For a threshold $u \in \mathcal{Y}$, assume that
\[
  p_u := P(Y>u)>0
\]
and define the finite tail measure
\begin{equation}\label{mu_mar}
  \mu_u(A) := P(Y>u,\, \mathbf{X}\in A),
  \qquad A\in\mathcal{B}(\mathcal{X}).
\end{equation}
Since $\mu_u\ll P_{\mathbf{X}}$, fix a Borel version of its Radon--Nikodym derivative,
\begin{equation}\label{eq:tail_score}
  \pi_u(\mathbf{x})
  :=
  \frac{\mathrm{d}\mu_u}{\mathrm{d}P_{\mathbf{X}}}(\mathbf{x})
  = P(Y>u\mid\mathbf{X}=\mathbf{x}),
\end{equation}
where the final equality holds $P_{\mathbf{X}}$-almost everywhere. Hence
\[
  \mu_u(A)=\int_A\pi_u(\mathbf{x})\,P_{\mathbf{X}}(\mathrm{d}\mathbf{x}).
\]
The normalized tail measure
\begin{equation}\label{eq:tail_law}
  Q_u(A)
  :=
  \frac{\mu_u(A)}{p_u}
  = P(\mathbf{X}\in A\mid Y>u)
\end{equation}
is the conditional covariate law given an exceedance. Thus $\mu_u(A)$ is joint tail mass, whereas $Q_u(A)$ is the fraction of exceedances whose covariates lie in $A$.

\begin{definition}[Risk halo]\label{mrh}
For $\alpha\in(0,1)$, a \emph{risk halo} is any solution, when one exists, of
\begin{equation}\label{eq:risk_halo_optimization}
  \mathscr{H}_{u,\alpha}
  \in
  \operatorname*{arg\,min}_{A\in\mathcal{B}(\mathcal{X})}
  \left\{
    P_{\mathbf{X}}(A):Q_u(A)\geq\alpha
  \right\}.
\end{equation}
Thus it is a covariate region of minimum $P_{\mathbf{X}}$-measure among those containing at least an $\alpha$-fraction of the exceedances above $u$.
\end{definition}

The criterion in \eqref{eq:risk_halo_optimization} is a set-function optimization problem in the sense of \cite{decarvalho2024}. The following condition gives existence, uniqueness, and a canonical superlevel-set representation.

\begin{proposition}[Risk-halo characterization]
\label{prop:risk_halo_characterization}
Let $\mathbf{X}_u\sim Q_u$, define $Z_u=\pi_u(\mathbf{X}_u)$, and suppose that the distribution of $Z_u$ is atom-free. If $F_{Z_u}^{-1}(q)=\inf\{t:F_{Z_u}(t)\geq q\}$ and
\[
  \tau_{u,\alpha}:=F_{Z_u}^{-1}(1-\alpha),
\]
then
\begin{equation}\label{suprep}
  \mathscr{H}_{u,\alpha}
  =
  \{\mathbf{x}\in\mathcal{X}:\pi_u(\mathbf{x})>\tau_{u,\alpha}\}
\end{equation}
is the unique risk halo up to $P_{\mathbf{X}}$-null sets, and
\begin{equation}\label{eq:risk_halo_exact_mass}
  Q_u(\mathscr{H}_{u,\alpha})=\alpha,
  \qquad
  \mu_u(\mathscr{H}_{u,\alpha})=\alpha p_u.
\end{equation}
Moreover, if $0<\alpha_1<\alpha_2<1$, the canonical representatives are nested:
\[
  \mathscr{H}_{u,\alpha_1}
  \subseteq
  \mathscr{H}_{u,\alpha_2}.
\]
In particular,
\[
  Q_u(\mathcal{X}\setminus\mathscr{H}_{u,\alpha})
  =1-\alpha\longrightarrow0
  \qquad\text{as }\alpha\uparrow1.
\]
If, in addition, $\mathcal{X}$ is closed and $Q_u$ has full support on $\mathcal{X}$, then
\[
  \mathscr{H}_{u,\alpha}
  \xrightarrow[\mathrm{PK}]{}
  \mathcal{X}
  \qquad\text{as }\alpha\uparrow1.
\]
\end{proposition}

Proposition~\ref{prop:risk_halo_characterization} also shows that $\alpha$ indexes a coherent path of tail-mass coverage: as $\alpha$ increases, the cutoff $\tau_{u,\alpha}$ does not increase, so configurations with progressively smaller values of $\pi_u$ may enter the halo without any previously included configuration being removed. The final conclusion adds a geometric counterpart to $Q_u(\mathcal{X}\setminus\mathscr{H}_{u,\alpha})\to0$. Painlev\'e--Kuratowski (PK) convergence is a canonical notion of convergence for sequences of sets \citep{molchanov2005}: the lower limit comprises points whose every neighborhood intersects all sufficiently late sets, whereas the upper limit comprises points whose every neighborhood intersects infinitely many sets; convergence means that these two limits coincide. Thus, under the stated support and closedness conditions, no neighborhood of a point in $\mathcal{X}$ remains disjoint from the halo as $\alpha\uparrow1$, and no limiting halo point lies outside $\mathcal{X}$.

\begin{remark}[Ties at the cutoff]
\label{rem:risk_halo_ties}
If $Z_u$ has atoms, choose $\tau$ such that
\[
  Q_u(\pi_u>\tau)\leq\alpha\leq Q_u(\pi_u\geq\tau).
\]
When the boundary level set can be split, a risk halo has the form
\[
  \{\pi_u>\tau\}\cup C,
  \qquad
  C\in\mathcal{B}(\mathcal{X}),
  \qquad C\subseteq\{\pi_u=\tau\},
\]
where $Q_u(C)=\alpha-Q_u(\pi_u>\tau)$ supplies the remaining tail mass; such a halo need not be unique. If the boundary contains indivisible atoms, the deterministic problem can instead be knapsack-type and need only satisfy the inequality in \eqref{eq:risk_halo_optimization}; randomized inclusion on the boundary restores the threshold rule with exact tail mass.
\end{remark}

\noindent The following parametric example illustrates the concept.
\begin{example}[Pareto regression, I]\label{tir}
  Suppose $Y \mid X = x \sim \mathrm{Pareto}(\alpha(x))$ and $X \sim \mathrm{Unif}[0, 1]$, where $\alpha(x) = \beta_1 + \beta_2 x$, $\beta_1 > 0$, $\beta_2 > 0$, and $\mathcal{X} = [0, 1]$. It follows that $P(Y > u \mid X = x) = u^{-(\beta_1 + \beta_2 x)}$ for $u \in \mathcal{Y} = (1, \infty)$. The law of total probability implies that
  \begin{equation}\label{margp}
    P(Y > u) = \int_{0}^1 u^{-\alpha(x)} \, \dif x = \frac{u^{-\beta_1}(1 - u^{-\beta_2})}{\beta_2 \log u}, 
  \end{equation}
and thus the resulting risk halo is obtained by solving for $x$ such that $\mu_u([0,x])\geq \alpha\mu_u([0,1])$. Hence, for $u > 1$ and $\alpha \in (0, 1)$, 
  \begin{equation}\label{rhpareto}
    \mathscr{H}_{u,\alpha} = [0, x_{u, \alpha}), \qquad
    x_{u, \alpha} 
    = -\frac{1}{\beta_2 \log u}\,
    \log\!\left(1 - \alpha(1 - u^{-\beta_2})\right). 
  \end{equation}
Figure~\ref{fig:pareto_halos} illustrates this with simulated data for $\alpha = 0.90$ and $u$ set to the marginal $0.95$ quantile of $Y$. For $\beta_1=1$ and $\beta_2=10$, this gives $u\approx2.35$ and $x_{u,\alpha}\approx0.27$. In this simulated sample, $0.93$ of the exceedances fall within the risk halo, close to the target fraction $0.90$.
\end{example}
\begin{figure}[t]
  \centering
  \includegraphics[width=0.45\textwidth]{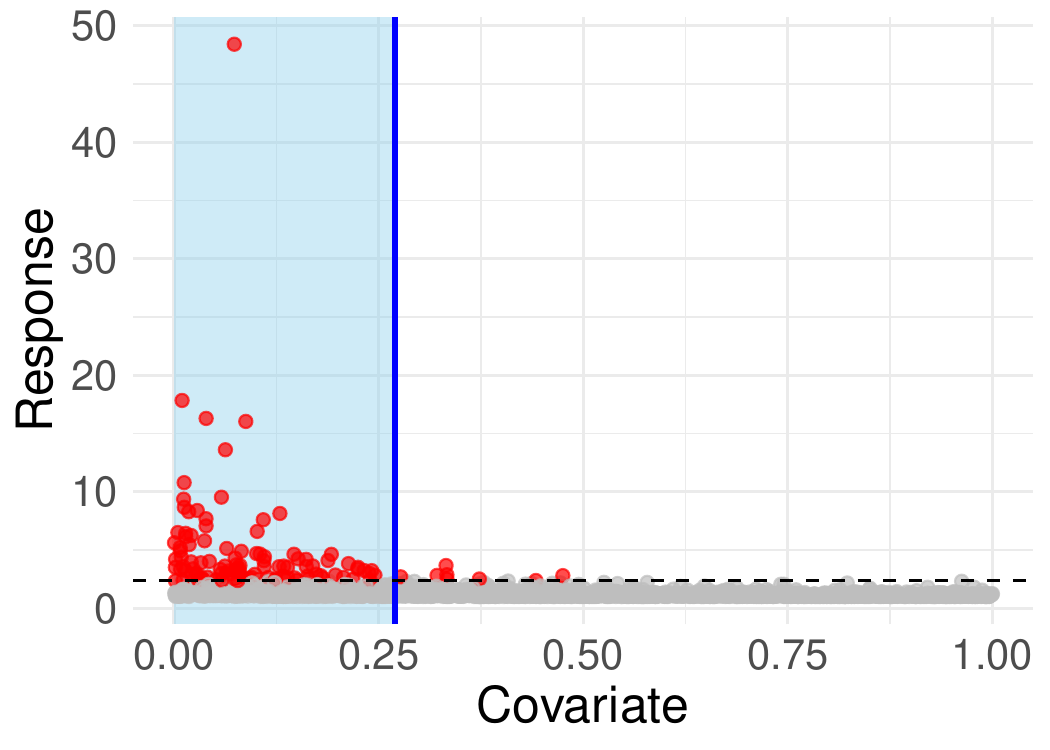} %
  \includegraphics[width=0.45\textwidth]{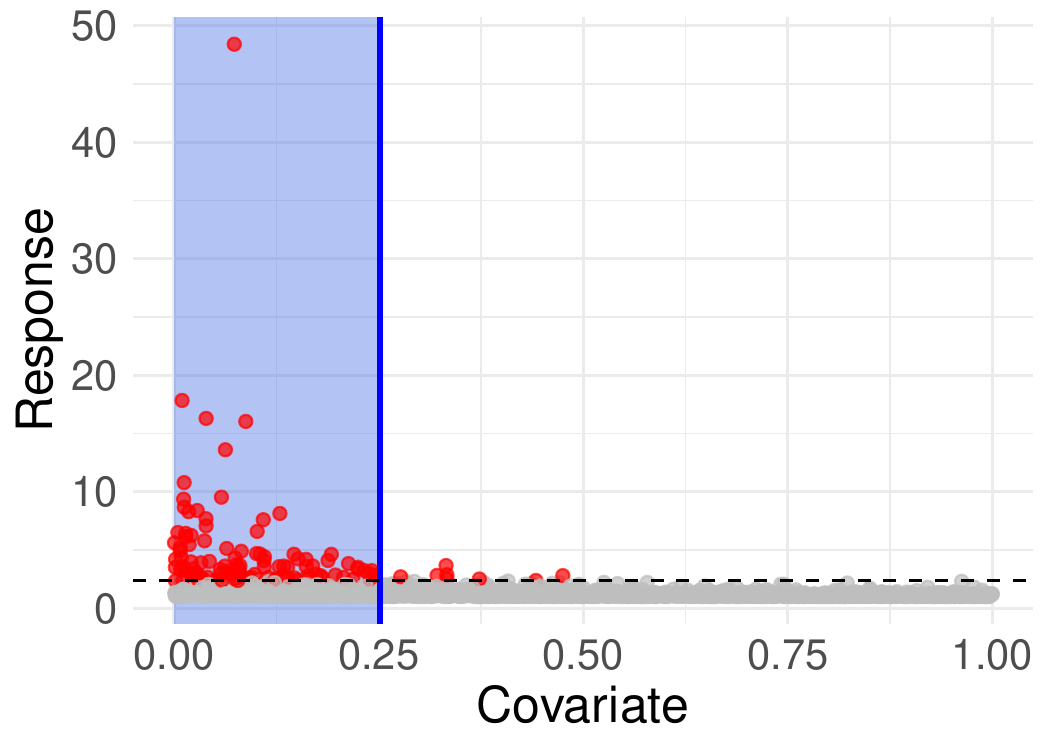}
  \caption{Left: True risk halo defining the covariate region capturing an $\alpha = 0.90$ fraction of the exceedances above the threshold $u$. Right: True favorable halo showing the region of covariate space in which the conditional exceedance probability exceeds its marginal level. Both panels use the same simulated data from Example~\ref{tir}, with $\beta_{1}=1$, $\beta_{2}=10$, and $u\approx2.35$, the marginal $0.95$ quantile of $Y$. The endpoints of the true halos in~\eqref{rhpareto} and~\eqref{favhalos}, rounded to two decimal places, are $0.27$ and $0.25$, respectively.}
  \label{fig:pareto_halos}
\end{figure}
Risk halos for bivariate responses can be defined analogously by modifying the tail measure in~\eqref{mu_mar}. For $(u, v) \in \mathcal{Y} \subseteq \mathbb{R}^2$, set
\begin{equation}\label{mu_mar2}
\begin{aligned}
  \mu_{u,v}(A)
  &:= P(Y_{1}>u,\, Y_{2}>v,\, \mathbf{X}\in A)\\
  &= \int_A P(Y_{1}>u,\, Y_{2}>v \mid \mathbf{X}=\mathbf{x})\,
     P_{\mathbf{X}}(\mathrm{d}\mathbf{x}),
  \qquad A\in\mathcal{B}(\mathcal{X}).
\end{aligned}
\end{equation}
Provided $p_{u,v}:=\mu_{u,v}(\mathcal{X})>0$, write
\[
  \pi_{u,v}:=\frac{\mathrm{d}\mu_{u,v}}{\mathrm{d}P_{\mathbf{X}}},
  \qquad
  Q_{u,v}:=\frac{\mu_{u,v}}{p_{u,v}}.
\]
Definition~\ref{mrh} then applies with $(u,v)$ in place of $u$, and Proposition~\ref{prop:risk_halo_characterization} applies under the corresponding atom-free condition on $Z_{u,v}=\pi_{u,v}(\mathbf{X}_{u,v})$ for $\mathbf{X}_{u,v}\sim Q_{u,v}$. In particular, under that condition, $\mathscr{H}_{u,v,\alpha}$ is a covariate region of minimum $P_{\mathbf{X}}$-measure among those containing at least an $\alpha$-fraction of the joint exceedances. The construction extends directly to higher-dimensional responses, with $\mathcal{Y} \subseteq \mathbb{R}^d$.
\subsection{Favorable halo}
The risk halo in Definition~\ref{mrh} is a concentration-based notion: it identifies covariate regions that contain at least an $\alpha$-fraction of all exceedances above $u$. A complementary notion identifies where the conditional exceedance probability is higher than its marginal level. For every measurable $C\subseteq\mathcal{X}$ with $P_{\mathbf{X}}(C)>0$, the version $\pi_u$ in \eqref{eq:tail_score} gives
\begin{equation}\label{eq:favorable_average}
  P(Y>u\mid\mathbf{X}\in C)-p_u
  =
  \frac{1}{P_{\mathbf{X}}(C)}
  \int_C\{\pi_u(\mathbf{x})-p_u\}\,
  P_{\mathbf{X}}(\mathrm{d}\mathbf{x}).
\end{equation}
Thus any positive-mass region contained, up to null sets, in $\{\pi_u>p_u\}$ is probability-increasing for the exceedance event. This is the covariate analogue of a favorable event \citep{chung1942,kramer2005}, stated without requiring singleton events $\{\mathbf{X}=\mathbf{x}\}$ to have positive probability.

\begin{definition}[Favorable halo]\label{fh}
  The favorable halo at threshold \(u\) is the set
  \begin{equation}\label{favhalos}
    \mathscr{F}_{u}
    = \{\mathbf{x}\in\mathcal{X}:\pi_u(\mathbf{x})>p_u\},
  \end{equation}
  understood up to $P_{\mathbf{X}}$-null sets. It is the covariate region in which a version of the conditional exceedance probability exceeds its marginal level.
\end{definition}

If $\pi_u$ is nonconstant $P_{\mathbf{X}}$-almost surely, then $P_{\mathbf{X}}(\mathscr{F}_u)>0$. In the no-signal case $\pi_u=p_u$ almost surely, the favorable halo is null and no favorable-halo-induced law is defined.

The same definition applies to a multivariate tail event. For example, with the bivariate notation above,
\[
  \mathscr{F}_{u,v}
  =\{\mathbf{x}\in\mathcal{X}:\pi_{u,v}(\mathbf{x})>p_{u,v}\}.
\]
More generally, for a componentwise threshold $\boldsymbol{u}$, one replaces $\pi_u$ and $p_u$ by $\pi_{\boldsymbol{u}}(\mathbf{x})=P(\mathbf{Y}>\boldsymbol{u}\mid\mathbf{X}=\mathbf{x})$ and $p_{\boldsymbol{u}}=P(\mathbf{Y}>\boldsymbol{u})$.

As can be seen by comparing \eqref{favhalos} with \eqref{suprep}, there are clear links between risk halos and favorable halos. However, while risk halos identify covariate regions that account for an $\alpha$-fraction of the exceedances above $u$, favorable halos are defined as the covariate regions that are comparatively more favorable for tail events.

\begin{example}[Pareto regression, II]\label{tir2}
  We consider again the Pareto regression framework from Example~\ref{tir}. It follows from \eqref{margp} that, for $u > 1$,
  \begin{equation*}
    \mathscr{F}_{u} = [0, x^*_{u}), \qquad 
  x_u^{*} 
  = -\frac{1}{\beta_2 \log u}
    \log\left(\frac{1 - u^{-\beta_2}}{\beta_2 \log u}\right).
  \end{equation*}
Figure~\ref{fig:pareto_halos} illustrates this with simulated data where 
$u$ is once again set to the marginal $0.95$ quantile of $Y$. For $\beta_1=1$ and $\beta_2=10$, this gives $u\approx2.35$ and $x_u^*\approx0.25$. The resulting region is in line with that provided by the risk halo.
\end{example}

\subsection{Halo-induced covariate laws}\label{sec:halo_distribution}
The preceding definitions identify subsets of the covariate space, but in practice such regions will be difficult, if not impossible, to represent directly when the covariate space has moderate or high dimensionality. Motivated by this, we associate each halo with the distribution of covariates inside the halo. This induced law is not merely a display device: it is the population object that converts a tail-relevant set into interpretable distributional summaries of the covariates that make up the set. If $B\in\mathcal{B}(\mathcal{X})$ is any halo with $P_{\mathbf{X}}(B)>0$, define the halo-induced covariate law
\begin{equation}\label{eq:halo_law}
  P_{\mathbf{X}}^{B}(A)
  =
  P(\mathbf{X}\in A\mid \mathbf{X}\in B)
  =
  \frac{P_{\mathbf{X}}(A\cap B)}{P_{\mathbf{X}}(B)},
  \qquad A\in\mathcal{B}(\mathcal{X}).
\end{equation}
For risk halos, $B=\mathscr{H}_{u,\alpha}$ or $B=\mathscr{H}_{u,v,\alpha}$; for favorable halos, $B=\mathscr{F}_u$ or $B=\mathscr{F}_{u,v}$. The law $P_{\mathbf{X}}^{B}$ answers a different question from the halo boundary. The boundary says which covariate values belong to the tail-relevant region, whereas $P_{\mathbf{X}}^{B}$ describes how covariate values are distributed within that region.

For the $j$th coordinate, let $P_{X_j}^{B}$ be the pushforward of $P_{\mathbf{X}}^{B}$ under the projection $\mathbf{x}\mapsto x_j$. Because $P_{X_j}^{B}\ll P_{X_j}$, the coordinate-specific enrichment ratio is
\begin{equation}\label{eq:true_halo_density}
  r_j^{B}(t)
  :=\frac{\mathrm{d}P_{X_j}^{B}}{\mathrm{d}P_{X_j}}(t)
  =
  \frac{P(\mathbf{X}\in B\mid X_j=t)}{P_{\mathbf{X}}(B)},
  \qquad P_{X_j}\text{-almost everywhere}.
\end{equation}
This definition applies to continuous, discrete, mixed, and cyclic covariates. If $P_{X_j}$ has density or mass function $f_j$ with respect to a reference measure $m_j$, then the halo-induced marginal has density or mass function
\begin{equation}\label{eq:halo_density_reference}
  f_j^{B}(t)=r_j^{B}(t)f_j(t)
  =\frac{P(\mathbf{X}\in B\mid X_j=t)f_j(t)}{P_{\mathbf{X}}(B)}.
\end{equation}
For example, $m_j$ may be Lebesgue measure for a continuous covariate, counting measure for a discrete covariate, or a suitable mixture of the two. If the full covariate vector has a density $f_{\mathbf{X}}$ with respect to a product reference measure, \eqref{eq:halo_density_reference} reduces to
\[
  f_j^{B}(t)
  =
  \frac{1}{P_{\mathbf{X}}(B)}
  \int_{\mathbb{R}^{p-1}}
  \mathbf{1}\{(t,\mathbf{x}_{-j})\in B\}
  f_{\mathbf{X}}(t,\mathbf{x}_{-j})\,m_{-j}(\mathrm{d}\mathbf{x}_{-j}).
\]
Here $f_{\mathbf{X}}$ is understood to be zero off $\mathcal{X}$.
Values with $r_j^{B}(t)>1$ are over-represented inside the halo relative to the marginal law of $X_j$, while values with $r_j^{B}(t)<1$ are under-represented. When densities exist, this is equivalent to comparing $f_j^{B}$ with $f_j$. In particular, when $X_j$ is marginally uniform, an approximately uniform halo density indicates that the halo is not concentrated in any particular part of that covariate's range; for a non-uniform covariate, the relevant comparison is not flatness, but similarity to the marginal density. These induced marginals should be read as covariate profiles of a halo, rather than as regression effects or marginal tail probabilities.

\begin{example}[Pareto regression, III]\label{tir3}
Return to the Pareto regression model in Examples~\ref{tir} and~\ref{tir2}, where $X\sim\mathrm{Unif}[0,1]$. In this case the marginal covariate density is $f_X(t)=1$ on $[0,1]$. Since the risk halo is $\mathscr{H}_{u,\alpha}=[0,x_{u,\alpha})$, its halo density is
\begin{equation*}
  f^{\mathscr{H}_{u,\alpha}}(t)
  =
  \frac{1}{x_{u,\alpha}}\,
  \mathbf{1}\{0\le t < x_{u,\alpha}\}.
\end{equation*}
Similarly, since the favorable halo is $\mathscr{F}_u=[0,x_u^*)$,
\begin{equation*}
  f^{\mathscr{F}_u}(t)
  =
  \frac{1}{x_u^*}\,
  \mathbf{1}\{0\le t < x_u^*\}.
\end{equation*}
Thus both halo densities are uniform after conditioning on the relevant halo, but they are not equal to the marginal covariate density on $[0,1]$. Their heights exceed one because the probability mass is concentrated on a shorter interval. Figure~\ref{fig:pareto_halo_densities} shows the corresponding true halo densities for $\beta_1=1$, $\beta_2=10$, $\alpha=0.90$, and $u\approx2.35$, the marginal $0.95$ quantile of $Y$. The risk- and favorable-halo densities have heights $3.71$ and $3.98$, respectively.
\end{example}

\begin{figure}[t]
  \centering
  \includegraphics[width=0.45\textwidth]{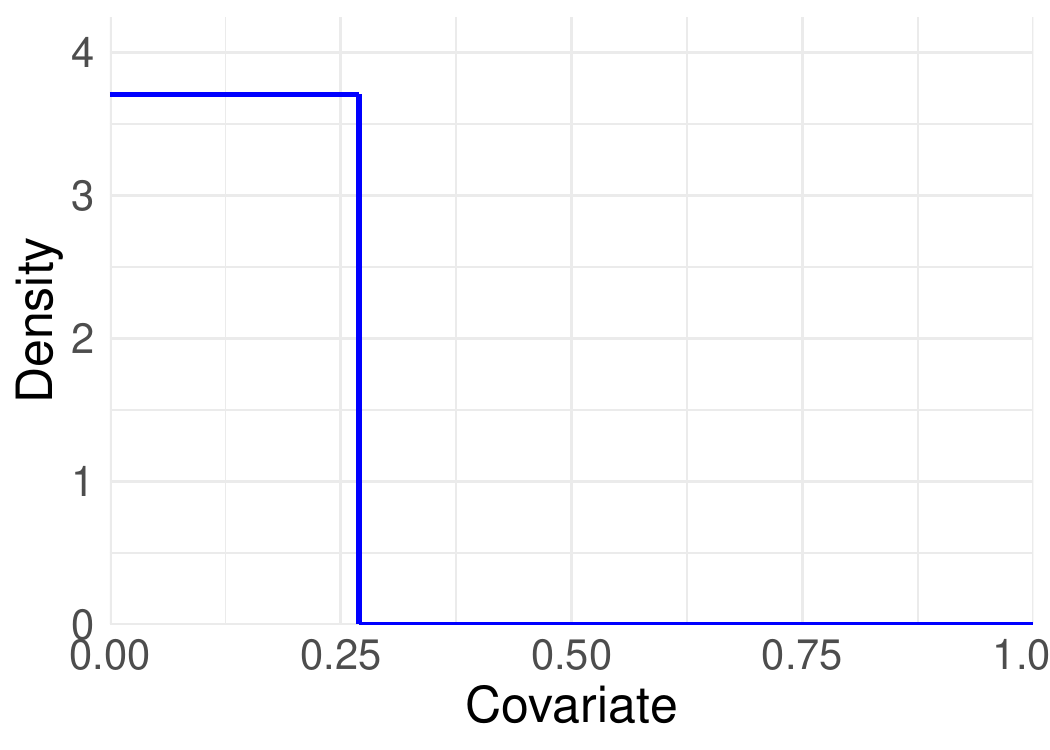}
  \includegraphics[width=0.45\textwidth]{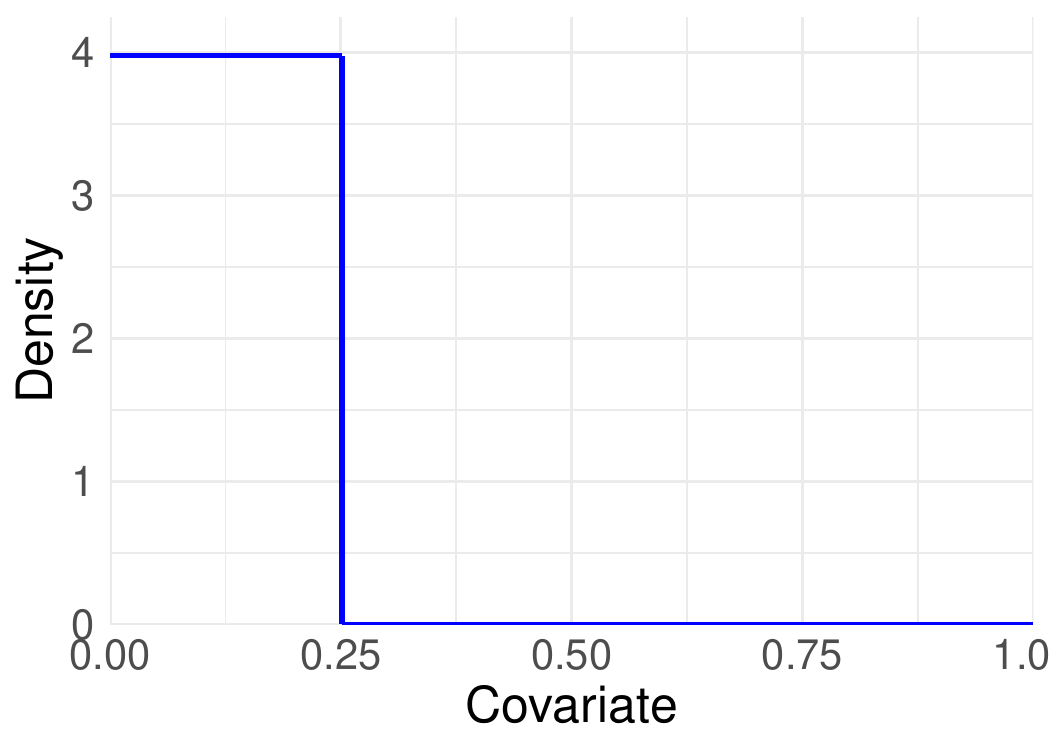}
  \caption{True halo densities for the Pareto regression example with $u\approx2.35$. Left: density induced by the risk halo, whose endpoint and height are $0.27$ and $3.71$. Right: density induced by the favorable halo, whose endpoint and height are $0.25$ and $3.98$. All reported values are rounded to two decimal places.}
  \label{fig:pareto_halo_densities}
\end{figure}

\begin{example}[EGPD regression]\label{egpd_halo_density_example}
The same calculation applies to the EGPD regression model used later for inference. Let $X$ have density $f_X$ on $\mathcal X\subseteq\mathbb R$, and suppose that, conditionally on $X=x$,
\begin{equation*}
  F(y\mid x)=G_{\kappa(x)}\{H(y;x)\},\qquad
  H(y;x)=1-\left(1+\frac{\xi(x)y}{\nu(x)}\right)^{-1/\xi(x)}.
\end{equation*}
For a fixed threshold $u$, define the conditional exceedance probability
\begin{equation*}
  \pi_u(x)=P(Y>u\mid X=x)=1-G_{\kappa(x)}\{H(u;x)\}.
\end{equation*}
Under the atom-free condition of Proposition~\ref{prop:risk_halo_characterization}, the risk halo is the superlevel set $B_R=\{x:\pi_u(x)>\tau_{u,\alpha}\}$, where $\tau_{u,\alpha}$ is chosen so that $\int_{B_R}\pi_u(x)f_X(x)\,\dif x=\alpha\int_{\mathcal X}\pi_u(x)f_X(x)\,\dif x$. The favorable halo is $B_F=\{x:\pi_u(x)>\bar\pi_u\}$, where $\bar\pi_u=\int_{\mathcal X}\pi_u(x)f_X(x)\,\dif x$. For either $B\in\{B_R,B_F\}$, the halo density is therefore
\begin{equation*}
  f^B(t)=
  \frac{f_X(t)\mathbf 1\{t\in B\}}
       {\int_B f_X(s)\,\dif s}.
\end{equation*}
Thus, once the EGPD exceedance probability has identified the halo, the induced halo density is simply the covariate density renormalized over that halo. In the special case $X\sim\mathrm{Unif}[0,1]$ and $\pi_u(x)$ is monotone increasing, both halos are upper intervals, say $B=[x_B,1]$, and $f^B(t)=(1-x_B)^{-1}\mathbf 1\{x_B\le t\le1\}$.
\end{example}

More generally, every halo and halo-induced law in this section is a functional of two ingredients: the conditional probability of the tail event and the covariate law $P_{\mathbf X}$. Section~\ref{model2} models the former and leaves the latter unrestricted, replacing it in the sample by the empirical covariate law.

\section{Learning halo-induced covariate laws from data}\label{model2}

The preceding section defined the population tail regions and associated them with the covariate laws they induce. We now connect those population objects to their fitted counterparts before specifying the conditional tail model and its posterior approximation.

\paragraph*{From population targets to fitted objects}
The connection with Section~\ref{model} is as follows. Let $e\in\mathcal E:=\{1,2,12\}$ index the two marginal tail events and the bivariate joint tail event, respectively. For fixed thresholds $u_1$ and $u_2$, the conditional marginal distributions and copula introduced below determine the tail scores
\begin{equation}
\begin{aligned}
  \pi^{(1)}(\x)
  &=P(Y_1>u_1\mid\X=\x)
    =1-F_1(u_1\mid\x),\\
  \pi^{(2)}(\x)
  &=P(Y_2>u_2\mid\X=\x)
    =1-F_2(u_2\mid\x),\\
  \pi^{(12)}(\x)
  &=P(Y_1>u_1,Y_2>u_2\mid\X=\x)\\
  &=1-F_1(u_1\mid\x)-F_2(u_2\mid\x)
    +C_{\x}\!\left\{F_1(u_1\mid\x),F_2(u_2\mid\x)\right\}.
\end{aligned}
\label{eq:model_tail_scores}
\end{equation}
Thus the EGPD marginals and the conditional copula enter the halo construction through the marginal tail scores in \eqref{eq:tail_score} and the joint tail measure in \eqref{mu_mar2}.

Given observed covariates $\x_1,\dots,\x_n$, we replace $P_{\mathbf X}$ by its empirical law. For any fitted tail score $\widehat\pi^{(e)}$, write $\widehat\pi_i^{(e)}=\widehat\pi^{(e)}(\x_i)$. Provided $\sum_{i=1}^n\widehat\pi_i^{(e)}>0$, the direct empirical counterparts of the covariate law, tail measure, marginal tail probability, and conditional covariate law in \eqref{mu_mar}--\eqref{eq:tail_law} are
\begin{equation}
\begin{aligned}
  \widehat P_{\mathbf X,n}
  &=\frac{1}{n}\sum_{i=1}^n\delta_{\x_i},\\
  \widehat\mu^{(e)}(A)
  &=\int_A\widehat\pi^{(e)}(\x)\,
    \widehat P_{\mathbf X,n}(\mathrm{d}\x)
    =\frac{1}{n}\sum_{i=1}^n
      \widehat\pi_i^{(e)}\mathbf 1\{\x_i\in A\},\\
  \widehat p^{(e)}
  &=\widehat\mu^{(e)}(\mathcal X)
    =\frac{1}{n}\sum_{i=1}^n\widehat\pi_i^{(e)},\\
  \widehat Q^{(e)}(A)
  &=\frac{\widehat\mu^{(e)}(A)}{\widehat p^{(e)}}
    =\frac{\sum_{i=1}^n\widehat\pi_i^{(e)}\mathbf 1\{\x_i\in A\}}
           {\sum_{i=1}^n\widehat\pi_i^{(e)}}.
\end{aligned}
\label{eq:empirical_tail_measures}
\end{equation}
On the indexed empirical sample, the quantities in \eqref{eq:empirical_tail_measures} lead to the risk-halo optimization and favorable-halo thresholding rules developed in Section~\ref{sec:halo_construction}. If $\widehat I^{(e)}\subseteq\{1,\dots,n\}$ is the index set selected by either rule, the corresponding selected-observation law is
\begin{equation}
  \widehat P_{\mathbf X,n}^{\,\widehat I^{(e)}}(A)
  =
  \frac{1}{|\widehat I^{(e)}|}
  \sum_{i\in\widehat I^{(e)}}\mathbf 1\{\x_i\in A\},
  \qquad A\in\mathcal B(\mathcal X),
\label{eq:empirical_halo_law}
\end{equation}
whenever $\widehat I^{(e)}$ is nonempty. This is a finite-sample representation on indexed observations. If the observed covariate vectors are distinct, it coincides with conditioning $\widehat P_{\mathbf X,n}$ on the selected covariate set and is therefore the direct empirical counterpart of \eqref{eq:halo_law}. With repeated covariate vectors at a cutoff, a deterministic covariate set must retain all copies of the repeated value; selecting only some copies is instead an index-level boundary split, analogous to the randomized treatment in Remark~\ref{rem:risk_halo_ties}. The coordinate-wise pushforwards of \eqref{eq:empirical_halo_law} are the empirical laws subsequently smoothed when a density representation is useful.

Finally, write $\bs\vartheta$ for the collection of marginal and copula model parameters, defined in detail below. For each draw $\bs\vartheta^{(s)}$ from the variational posterior, we evaluate the scores in \eqref{eq:model_tail_scores}, form the empirical quantities in \eqref{eq:empirical_tail_measures}, select the corresponding halo index set, and construct its induced law using \eqref{eq:empirical_halo_law}. Repeating these steps draw by draw propagates uncertainty from the fitted conditional tail model to the halo and its induced covariate law. The next subsections specify the conditional EGPD marginals, the covariate-dependent copula, the structured additive spike-and-slab prior, and the variational approximation that supply the draws $\bs\vartheta^{(s)}$. Section~\ref{sec:halo_construction} then gives the selection and smoothing details.

\subsection{Marginal model}

Let $Y$ denote a positive response observed together with covariates $\X=\x$. We model the conditional marginal distribution of $Y$ given $\X=\x$ using the extended generalized Pareto distribution (EGPD); see \citet{naveau2016} and \citet{decarvalho2021a}. The conditional cumulative distribution function is
\begin{align}
F(y \mid \x) &= G_{\kappa(\x)}\bigl(H(y;\x)\bigr),
\label{eq:egpd_cond_cdf}
\end{align}
where $G_{\kappa(\x)}: [0,1] \to [0,1]$ is a carrier function parametrized by $\kappa(\x)$, and
\begin{align*}
H(y;\x) &= 1 - \left(1 + \frac{\xi(\x)y}{\nu(\x)}\right)^{-1/\xi(\x)},
\end{align*}
defined on $\{y \in (0,\infty): 1+\xi(\x)y/\nu(\x)>0\}$. Here, $\nu(\x)>0$ is a scale parameter and $\xi(\x)\in\mathbb{R}$ is the shape parameter; the case $\xi(\x)=0$ is understood as the limit $\xi(\x)\to0$. The shape parameter $\xi(\x)$ is the extreme value index and controls the rate of decay of the tail.

Differentiating \eqref{eq:egpd_cond_cdf}, the corresponding conditional density is
\begin{align*}
f(y \mid \x) &= g_{\kappa(\x)}\bigl(H(y;\x)\bigr) h(y;\x),
\end{align*}
where $g_{\kappa(\x)}$ is the derivative of the carrier function $G_{\kappa(\x)}$ and
\begin{align*}
h(y;\x) &= \frac{1}{\nu(\x)}\left(1 + \frac{\xi(\x)y}{\nu(\x)}\right)^{-1/\xi(\x) - 1}.
\end{align*}

To incorporate covariate information, we let each distributional parameter $p \in \{\nu,\xi,\kappa\}$ depend on $\x$ through a structured additive predictor. More precisely, we write
\begin{align}
p(\x;\bs\beta_p) &= h_p\{\eta_p(\x)\}, \label{eq:marginal_response}
\end{align}
where $h_p$ is a suitable response function and
\begin{align}
\eta_p(\x) &= \sum_{j=1}^{J_p} f_{p,j}(x_j,\bs\beta_{p,j}). \label{eq:linear_term}
\end{align}
Here, $\bs\beta_p$ collects all regression coefficients entering the predictor for parameter $p$, whereas $\bs\beta_{p,j}$ denotes the subvector associated with the $j$th covariate effect. Thus the notation in \eqref{eq:marginal_response} refers to the full parameter-specific coefficient vector, while \eqref{eq:linear_term} displays its effect-wise decomposition. The functions $f_{p,j}$ may represent effects of different types, such as linear or smooth effects. In this way, the covariates affect the conditional marginal distribution by acting directly on its distributional parameters.

\subsection{Multivariate dependence model}

For a $D$-variate response $\mathbf{Y}=(Y_1,\dots,Y_D)$ observed with covariates $\X=\x$, each margin is assigned its own conditional EGPD model. For $d \in \{1,\dots,D\}$,
\begin{align}
F_d(y \mid \x) &= G_{\kappa_d(\x)}\bigl(H_d(y;\x)\bigr), \label{eq:egpd_cond_cdf_d}
\end{align}
with distributional parameters linked to $\x$ through structured additive predictors as in \eqref{eq:marginal_response}--\eqref{eq:linear_term}. Collecting all marginal and copula parameters, we define the covariate-dependent distributional parameter vector
\begin{align}
\boldsymbol{\varphi}(\x) = \bigl(\nu_1(\x),\xi_1(\x),\kappa_1(\x),\dots,\nu_D(\x),\xi_D(\x),\kappa_D(\x),\boldsymbol{\lambda}(\x)\bigr), \label{eq:full_parameter_vector}
\end{align}
where $\boldsymbol{\lambda}(\x)$ denotes the copula parameter vector. Thus, the marginal modeling step yields $D$ conditional distribution functions $F_1(\cdot \mid \x),\dots,F_D(\cdot \mid \x)$, one for each component of $\mathbf{Y}$, and these form the marginal part of $\boldsymbol{\varphi}(\x)$.

To combine these conditional marginals into a valid joint model, we adopt a copula-based construction. By Sklar’s theorem, any joint distribution with continuous marginals can be decomposed into its marginals and a copula. In the conditional setting, Patton’s extension of Sklar’s theorem shows that, provided the same conditioning information is used in both layers, the conditional joint distribution can be written in terms of the conditional marginals and a conditional copula; see \cite{patton2006modelling}. Accordingly, we define
\begin{align}
F(y_1,\dots,y_D \mid \x) &= C_{\x}\bigl(F_1(y_1 \mid \x),\dots,F_D(y_D \mid \x)\bigr), \label{eq:conditional_copula}
\end{align}
where $C_{\x}(\cdot)$ is a conditional copula. Hence, the conditional EGPD marginals from \eqref{eq:egpd_cond_cdf_d} enter the joint model only through their probability integral transforms, while the copula captures the remaining conditional dependence. In this way, \eqref{eq:conditional_copula} complements the marginal EGPD specification by turning the covariate-dependent parameter vector $\boldsymbol{\varphi}(\x)$ in \eqref{eq:full_parameter_vector} into a coherent conditional joint model.

We use parametric copulas and allow their parameters to vary with covariates through the same structured additive specification adopted for the marginal parameters. This lets the dependence structure change with external conditions while preserving interpretability, since many copula families admit parameters directly linked to Kendall’s $\tau$, tail dependence coefficients, or related measures of association.

For instance, consider the BB1 copula, parametrized by $\theta > 0$ and $\delta > 1$. We let both parameters vary with $\x$ according to
\begin{align*}
\theta(\x) = \exp\{\eta_{\theta}(\x)\}, \quad \delta(\x) = \exp\{\eta_{\delta}(\x)\} + 1,
\end{align*}
where $\eta_{\theta}(\x)$ and $\eta_{\delta}(\x)$ are defined as in \eqref{eq:linear_term} with $p \in \{\theta,\delta\}$. Hence, the copula parameter vector becomes $\boldsymbol{\lambda}(\x)=(\theta(\x),\delta(\x))$, so that both the marginal distributions and the dependence structure are functions of the same covariates: the EGPD layer governs the conditional tail behavior of each component, while the copula layer governs how these conditional extremes co-occur.

\subsection{Spike-and-Slab Model in Structured Additive Distributional Regression under Stochastic Variational Inference}\label{sec:spike-slab-model}

Our effect-selection strategy builds on the Bayesian shrinkage perspective of \citet{decarvalho2021a}, who developed an extreme-value Bayesian Lasso for conditional tail regression. Here, we extend this idea to nonlinear structured additive effects, multiple distributional parameters, and multivariate dependence through covariate-dependent copulas.

To enable automatic and interpretable selection of nonlinear covariate effects in the additive predictors of our marginal and dependence models, we employ the Normal--Beta Prime Spike-and-Slab (NBPSS) prior of \citet{klein_2021}, adapted for inference via stochastic variational inference (SVI). This prior allows the model to distinguish between relevant and negligible smooth effects by adaptively shrinking some functions toward zero.

Let $d$ denote the model component, that is, one of the $D$ marginal components or the copula component, let $p$ denote a distributional parameter within component $d$, let $j=1,\dots,J_{d,p}$ index the effects entering the predictor of parameter $p$, and let $i=1,\dots,n$ index the observations. For observation $i$, the $j$th smooth effect entering parameter $p$ of component $d$ is written as $f_{i,d,p,j}$ and decomposed into a penalized and an unpenalized component:
\begin{align*}
f_{i,d,p,j}
=
f^{\text{unpenalized}}_{i,d,p,j}
+
f^{\text{penalized}}_{i,d,p,j},
\end{align*}
and the penalized coefficients $\bs\beta^{\text{penalized}}_{d,p,j}$ follow a Gaussian prior conditional on a local variance parameter $\tau^2_{d,p,j}$,
\begin{align*}
p(\bs\beta^{\text{penalized}}_{d,p,j}\mid\tau^2_{d,p,j})
\propto
\exp\left\{
-\frac{1}{2\tau^2_{d,p,j}}
\left(\bs\beta^{\text{penalized}}_{d,p,j}\right)^\top
\mathbf{K}^{\text{penalized}}_{d,p,j}
\bs\beta^{\text{penalized}}_{d,p,j}
\right\}.
\end{align*}
The hierarchical structure of the NBPSS prior is preserved, using a shape--rate parameterization for the Gamma distribution,
\begin{align*}
\tau^2_{d,p,j}\mid\rho_{d,p,j},\psi^2_{d,p,j}
&\sim
\mathrm{Gamma}\!\left(
\frac{1}{2},
\frac{1}{2r_{d,p,j}(\rho_{d,p,j})\psi^2_{d,p,j}}
\right),\\
\rho_{d,p,j}\mid\omega_{d,p,j}
&\sim
\mathrm{Bernoulli}(\omega_{d,p,j}),\\
\psi^2_{d,p,j}
&\sim
\mathrm{Inverse\text{-}Gamma}(a_{d,p,j},b_{d,p,j}),\\
\omega_{d,p,j}
&\sim
\mathrm{Beta}(a_{0,d,p,j},b_{0,d,p,j}),
\end{align*}
with
\begin{align}
r_{d,p,j}(\rho_{d,p,j})=
\begin{cases}
r_{d,p,j}, & \text{if } \rho_{d,p,j}=0,\\
1, & \text{if } \rho_{d,p,j}=1,
\end{cases}
\label{eq:r_pj}
\end{align}
where $r_{d,p,j}>0$ is a small fixed value that enforces shrinkage under exclusion. The scale parameter $\psi^2_{d,p,j}$ determines the prior expectation of $\tau^2_{d,p,j}$, which is $\psi^2_{d,p,j}$ when $\rho_{d,p,j}=1$, and $r_{d,p,j}\psi^2_{d,p,j}$ when $\rho_{d,p,j}=0$. The indicator $\rho_{d,p,j}$ therefore determines whether the penalized nonlinear effect is assigned to the spike or slab component, inducing strong shrinkage toward zero under exclusion. The parameter $\omega_{d,p,j}$ denotes the prior inclusion probability, and the remaining terms $a_{d,p,j}, b_{d,p,j}, a_{0,d,p,j}, b_{0,d,p,j}$, and $r_{d,p,j}$ are hyperparameters of the spike-and-slab prior.

\paragraph*{Relaxed Bernoulli Approximation}
In \citet{klein_2021}, the indicator $\rho_{d,p,j}$ is assigned a Bernoulli prior. Since this prior is not differentiable, we replace it by a Relaxed Bernoulli (or Concrete) distribution to enable estimation by stochastic variational inference; see Section~\ref{section:SVIModel}. Specifically, we use the continuous and reparameterizable approximation
\begin{align*}
\tilde{\rho}_{d,p,j}=\operatorname{sigmoid}\!\left(\frac{\log\omega_{d,p,j}-\log(1-\omega_{d,p,j})+\log u-\log(1-u)}{T_{d,p,j}}\right), \qquad u\sim\mathrm{Unif}(0,1),
\end{align*}
where $T_{d,p,j}>0$ is an effect-specific temperature parameter controlling the sharpness of the relaxation. Rather than fixing $T_{d,p,j}$ a priori, we treat it as an additional hyperparameter and estimate it jointly with the remaining spike-and-slab quantities. As $T_{d,p,j}\to 0$, the Relaxed Bernoulli converges to the original Bernoulli distribution. This construction allows gradients to propagate through $\tilde{\rho}_{d,p,j}$ and therefore enables fully differentiable optimization of the variational objective. Finally, the expected inclusion probability $\mathbb{E}_q[\tilde{\rho}_{d,p,j}]$ is used to assess whether an effect is active in the model: values close to one indicate relevance, whereas values close to zero imply strong shrinkage of the corresponding spline coefficients.

\subsection{Variational Approximation under Stochastic Variational Inference}\label{section:SVIModel}

Posterior inference is performed via stochastic variational inference. Let $\mathcal{D}=\{1,\dots,D,c\}$ denote the set of model components, where $c$ stands for the copula component. For each $d\in\mathcal{D}$, let $\mathcal{P}_d$ denote the set of distributional parameters associated with component $d$. For each $d\in\mathcal{D}$, $p\in\mathcal{P}_d$, and $j=1,\dots,J_{d,p}$, define
\begin{align*}
\bs\zeta_{d,p,j}=\left(\logit\,\omega_{d,p,j},\,\logit\,\rho_{d,p,j},\,\log\psi^2_{d,p,j},\,\log\tau^2_{d,p,j},\,\log T_{d,p,j}\right)^\T.
\end{align*}
We then collect all unknown quantities into
\begin{align*}
\bs\vartheta=\left(\{\bs\beta_d\}_{d\in\mathcal{D}},\,\{\bs\zeta_{d,p,j}\}_{d\in\mathcal{D},\,p\in\mathcal{P}_d,\,j=1,\dots,J_{d,p}}\right),
\end{align*}
where $\bs\beta_d$ collects all regression coefficients associated with component $d$. We approximate the posterior $p(\bs\vartheta\mid\mathbf{y})$ by a tractable variational distribution $q_{\bs\phi}(\bs\vartheta)$ within a prescribed family $\mathcal{Q}$. The variational parameters $\bs\phi$ are obtained by maximizing the evidence lower bound
\begin{align*}
\mathcal{L}(\bs\phi)
=
\mathbb{E}_{q_{\bs\phi}(\bs\vartheta)}
\!\left[
\log p(\mathbf{y},\bs\vartheta)
-
\log q_{\bs\phi}(\bs\vartheta)
\right].
\end{align*}
Optimization is performed using reparameterization gradients and Adam \citep{kingma2014adam}.

To exploit the structure of the model, we use a blockwise variational factorization over the regression-coefficient blocks and the effect-specific hyperparameter blocks.

For each component $d$, the regression coefficients are modeled jointly by a multivariate Gaussian variational distribution,
\begin{align}
q_{\bs\phi}(\bs\beta_d)=\mathrm{N}\!\left(\bs\mu_{\beta,d},\bs\Sigma_{\beta,d}\right),
\label{eq:q_beta}
\end{align}
which allows posterior dependence between spline coefficients within the same marginal or copula block. Likewise, for each smooth effect we specify
\begin{align}
q_{\bs\phi}(\bs\zeta_{d,p,j})=\mathrm{N}\!\left(\bs\mu_{\zeta,d,p,j},\bs\Sigma_{\zeta,d,p,j}\right),
\label{eq:q_zeta}
\end{align}
allowing posterior dependence among the five local spike-and-slab quantities within each effect.

The resulting variational approximation is therefore
\begin{align*}
q_{\bs\phi}(\bs\vartheta)=\prod_{d\in\mathcal{D}}\left[q_{\bs\phi}(\bs\beta_d)\prod_{p\in\mathcal{P}_d}\prod_{j=1}^{J_{d,p}}q_{\bs\phi}(\bs\zeta_{d,p,j})\right].
\end{align*}
This factorization preserves dependence within each regression-coefficient block $\bs\beta_d$ and within each effect-specific hyperparameter block $\bs\zeta_{d,p,j}$, while assuming independence across these blocks.

\subsection{From fitted exceedance probabilities to halo-induced laws}\label{sec:halo_construction}

We now apply \eqref{eq:empirical_tail_measures} and \eqref{eq:empirical_halo_law} draw by draw. The primary empirical target is not the literal boundary of a continuous set, but the covariate law induced by the fitted halo. We restrict attention to the bivariate case $D=2$ and use the event index $e\in\mathcal E$ introduced above. At a single fitted parameter value, evaluating \eqref{eq:model_tail_scores} at $\x_i$ gives $\widehat\pi_i^{(e)}$. At posterior draw $s$, the same evaluation gives $\widehat\pi_i^{(e,s)}$, and the full construction below is repeated draw by draw. In the simulation study in \autoref{sec:simulation}, replacing fitted parameters by the data-generating parameters gives the benchmark representation on the same indexed covariate sample.

\paragraph*{Risk halos}

On the indexed sample, substituting the empirical quantities in \eqref{eq:empirical_tail_measures} into the population optimization in \eqref{eq:risk_halo_optimization} gives
\begin{equation}
  \widehat I_{\mathrm{R},\alpha}^{(e)}
  \in
  \operatorname*{arg\,min}_{I\subseteq\{1,\dots,n\}}
  \left\{
    \frac{|I|}{n}:
    \frac{\sum_{i\in I}\widehat\pi_i^{(e)}}
         {\sum_{i=1}^n\widehat\pi_i^{(e)}}
    \geq\alpha
  \right\}.
\label{eq:empirical_risk_halo_optimization}
\end{equation}
Because every observation has the same empirical mass $1/n$, a canonical solution of \eqref{eq:empirical_risk_halo_optimization} is obtained by ranking the observations according to their fitted tail scores. Sort the indices so that
\begin{align*}
\widehat{\pi}^{(e)}_{i_1} \ge \dots \ge \widehat{\pi}^{(e)}_{i_n},
\end{align*}
define
\begin{align*}
c_k^{(e)} &= \frac{\sum_{\ell=1}^k \widehat{\pi}^{(e)}_{i_\ell}}{\sum_{\ell=1}^n \widehat{\pi}^{(e)}_{i_\ell}}, \qquad k = 1,\dots,n.
\end{align*}
If $k^{\ast}_e = \min\{k : c_k^{(e)} \ge \alpha\}$, take
\begin{align*}
\widehat I^{(e)}_{\mathrm{R},\alpha}
&=\{i_1,\dots,i_{k^{\ast}_e}\},
&
\widehat B^{(e)}_{\mathrm{R},\alpha}
&=\{\x_i:i\in\widehat I^{(e)}_{\mathrm{R},\alpha}\}.
\end{align*}
For the index representation, ties at the cutoff may be broken arbitrarily; a deterministic covariate-set representation instead groups repeated covariate vectors as described after \eqref{eq:empirical_halo_law}. Because the empirical law is atomic, the captured $\widehat Q^{(e)}$-mass can exceed $\alpha$ at the final included observation, consistently with Remark~\ref{rem:risk_halo_ties}.

\paragraph*{Favorable halos}

The empirical marginal tail probability in \eqref{eq:empirical_tail_measures} gives the fitted favorable-halo threshold directly. Thus
\begin{equation}
  \widehat p^{(e)}
  =\frac{1}{n}\sum_{i=1}^n\widehat\pi_i^{(e)},
  \qquad
  \widehat I_{\mathrm F}^{(e)}
  =\{i:\widehat\pi_i^{(e)}>\widehat p^{(e)}\},
  \qquad
  \widehat B_{\mathrm F}^{(e)}
  =\{\x_i:i\in\widehat I_{\mathrm F}^{(e)}\}.
\label{eq:empirical_favorable_halo}
\end{equation}
Equation~\eqref{eq:empirical_favorable_halo} is therefore the direct empirical counterpart of \eqref{favhalos}. In the simulation study, the corresponding benchmark index sets are obtained by replacing $\widehat\pi_i^{(e)}$ with the exact scores from the data-generating model.

\paragraph*{Estimating halo densities}

Let $\widehat I^{(e,s)}$ denote either posterior-draw index set $\widehat I_{\mathrm{R},\alpha}^{(e,s)}$ or $\widehat I_{\mathrm F}^{(e,s)}$. When this set is nonempty, substituting it into \eqref{eq:empirical_halo_law} gives the draw-specific empirical halo law. For a continuous coordinate $X_j$, we smooth its pushforward using
\begin{equation}
  \widehat f_j^{(e,s)}(x;h_j)
  =\frac{1}{h_j|\widehat I^{(e,s)}|}
   \sum_{i\in\widehat I^{(e,s)}}
   K\!\left(\frac{x-x_{i,j}}{h_j}\right).
\label{eq:kde_draw}
\end{equation}
Here $K$ is the Gaussian kernel and $h_j>0$ is a bandwidth chosen automatically from the observed covariate values. Pointwise summaries of the collection of curves in \eqref{eq:kde_draw} describe variational-posterior uncertainty in the smoothed finite-sample halo law, conditional on the observed covariates and on the fitted halo being nonempty. If a favorable-halo draw is empty, no induced law exists for that draw; such draws are recorded separately and excluded from the density summaries.

For a single membership-weighted summary, let $w_i^{(e)}$ be the variational-posterior probability that observation $i$ belongs to the fitted halo, estimated by averaging the draw-specific membership indicators. We use the normalized weighted kernel density
\begin{align}
\widehat f_j^{(e)}(x;h_j)
=\frac{1}{h_j\sum_{i=1}^n w_i^{(e)}}
 \sum_{i=1}^n w_i^{(e)}
 K\!\left(\frac{x-x_{i,j}}{h_j}\right),
\label{eq:kde}
\end{align}
provided at least one weight is positive. This curve summarizes posterior halo membership, but it need not equal the pointwise posterior mean of \eqref{eq:kde_draw}, because the normalizing halo size can vary across draws. We estimate the overall sample density of $X_j$ using all observations with equal weights. Comparing the fitted halo-density summaries with the full-sample density estimates the contrast between $f_j^B$ and the marginal density of $X_j$ in Section~\ref{sec:halo_distribution}. Peaks relative to the overall covariate density indicate values that are over-represented inside the fitted halo; troughs indicate values that are under-represented. If the curves are close, the fitted halo has little visible structure along that covariate. This construction propagates variational-posterior uncertainty in the conditional tail scores while conditioning on the observed covariates; it does not posit a Bayesian model for $P_{\mathbf X}$.

\section{Simulation study}\label{sec:simulation}

\subsection{Simulation design, implemented model, and metrics of performance}
Here, we investigate the finite-sample performance of the proposed structured additive distributional regression model and assess how accurately it recovers the finite-sample halo representations used to construct the true risk and favorable halo densities. The main numerical findings are presented in Section~\ref{mainsim}; for now, we focus on the simulation design, summarize the implemented version of the proposed model, and outline the metrics used to assess performance.

We simulate a bivariate response $Y = (Y_1, Y_2)$ given a $J$-dimensional covariate vector $\X = (X_1,\dots,X_J)$. Each covariate is sampled independently from a uniform distribution on one of three intervals,
\begin{align*}
X_j \sim \mathrm{Unif}(a_j,b_j), \qquad (a_j,b_j) \in \{[-2\pi,0],[-\pi,\pi],[0,2\pi]\},
\end{align*}
with the interval for each $X_j$ chosen at random at the beginning of the experiment and then kept fixed across replications. Conditionally on $\X=\x$, the marginal distributions of $Y_1$ and $Y_2$ are specified as EGPDs with parameters $\bigl(\nu_d(\x),\xi_d(\x),\kappa_d(\x)\bigr)$, for $d=1,2$, where the carrier function is taken to be the power carrier $G_\kappa(v)=v^\kappa$ and the positivity constraint on $\kappa_d(\x)$ is enforced through a softplus link. The joint distribution of $(Y_1,Y_2)$ is constructed via a BB1 copula with parameters $\theta(\x)$ and $\delta(\x)$, providing positive tail dependence and allowing the strength of dependence to vary over $\x$. The true regression functions for all EGPD and copula parameters are fixed as
\begin{equation}
\begin{aligned}
\nu_{1}(\x)
&= \operatorname{softplus}\bigl(0.29 - 1.3\sin x_{3} + 0.7\cos x_{7}\bigr), \\
\xi_{1}(\x)
&= \operatorname{softplus}\bigl(0.41 - 0.9\cos x_{5} - 1.2\cos x_{6}\bigr) - 0.5, \\
\kappa_{1}(\x)
&= \operatorname{softplus}\bigl(0.63 + \cos x_{6} - 0.7\cos x_{9}\bigr), \\[0.4em]
\nu_{2}(\x)
&= \operatorname{softplus}\bigl(0.53 - 1.3\sin x_{2} - 0.8\sin x_{9}\bigr), \\
\xi_{2}(\x)
&= \operatorname{softplus}\bigl(-0.36 + \cos x_{6} - 0.8\sin x_{9}\bigr) - 0.5, \\
\kappa_{2}(\x)
&= \operatorname{softplus}\bigl(-0.48 + 0.8\sin x_{1} + 0.9\sin x_{9}\bigr), \\[0.4em]
\theta(\x)
&= \operatorname{softplus}\bigl(-0.66 + 0.8\sin x_{1} + \cos x_{9}\bigr), \\
\delta(\x)
&= \operatorname{softplus}\bigl(-0.64 + 0.7\cos x_{2} - 0.8\sin x_{3}\bigr) + 1.
\end{aligned}
\label{eq:true-regression-functions}
\end{equation}

For each parameter, only two covariates are truly relevant; all remaining covariates have no effect and therefore constitute noise for the variable-selection mechanism.

We consider three sample sizes, $n \in \{1000, 2000, 4000\}$. For each $n$ we generate $100$ independent datasets from the above model. For each simulated dataset we fit the full multivariate structured additive distributional regression model introduced in Section~\ref{sec:spike-slab-model}, including spike-and-slab priors on all nonlinear effects for the marginal and copula parameters. Each covariate effect is modeled by a univariate P-spline with $10$ interior knots; we apply the reparameterization presented in~\ref{sec:spline_constraints} to separate the unpenalized linear part from the nonlinear component that carries the spike-and-slab prior.

Posterior inference is performed by stochastic variational inference as in Section~\ref{section:SVIModel}. The variational family factorizes blockwise over marginal and copula components and within each component over regression coefficients and local spike-and-slab hyperparameters. We use $S=64$ Monte Carlo samples per iteration and run SVI for $20{,}000$ epochs. Optimization is carried out with Adam with a learning rate $0.001$, using one optimizer for all regression coefficients and a second optimizer for all hyperparameters.

Hyperprior choices largely mirror those in the simulation study of \citet{klein_2021}, with adaptations to our setting. For each smooth effect we use
\begin{align*}
  \psi^{2}_{d,p,j}
  &\sim \mathrm{Inverse\text{-}Gamma}(5,20), \\
  \omega_{d,p,j}
  &\sim \mathrm{Beta}(1,1), \\
  \log T_{d,p,j}
  &\sim \mathrm{Normal}(-1.5,0.5^{2}),
\end{align*}
and we fix $r_{d,p,j} = 10^{-4}$ for all components in the spike part of the prior, giving a strong separation between spike and slab. The inverse-gamma scale parameter is kept at $20$ to provide a slightly wider slab, which we found helpful for moderate sample sizes. The prior on $\log T_{d,p,j}$ is chosen so that the corresponding temperature $T_{d,p,j}$ is typically small, thereby making the Relaxed Bernoulli distribution close to the original Bernoulli inclusion prior. In particular, small values of $T_{d,p,j}$ force the relaxed inclusion variable to concentrate near $0$ and $1$, so that the continuous approximation closely mimics a discrete Bernoulli prior.

\begin{table}
\centering
\begin{tabular}{@{}cccccc@{}}
\toprule
\textbf{Response} & \textbf{Parameter} & \textbf{Covariate}
& $n = 1000$ & $n = 2000$ & $n = 4000$ \\ \midrule
\multirow{9}{*}{$y_1$}
  & \multirow{3}{*}{$\nu$}
      & $x_3$                              & 0.99 & 0.99 & 0.99 \\
  &   & $x_7$                              & 0.86 & 0.85 & 0.86 \\
  &   & $x_j,\ j \notin \{3,7\}$          & 0.03 & 0.01 & 0.01 \\ \cmidrule(l){2-6}
  & \multirow{3}{*}{$\xi$}
      & $x_5$                              & 0.94 & 0.94 & 0.95 \\
  &   & $x_6$                              & 0.99 & 0.99 & 0.99 \\
  &   & $x_j,\ j \notin \{5,6\}$          & 0.09 & 0.03 & 0.01 \\ \cmidrule(l){2-6}
  & \multirow{3}{*}{$\kappa$}
      & $x_6$                              & 0.98 & 0.97 & 0.97 \\
  &   & $x_9$                              & 0.78 & 0.77 & 0.75 \\
  &   & $x_j,\ j \notin \{6,9\}$          & 0.03 & 0.02 & 0.01 \\ \midrule
\multirow{9}{*}{$y_2$}
  & \multirow{3}{*}{$\nu$}
      & $x_2$                              & 0.99 & 0.99 & 0.99 \\
  &   & $x_9$                              & 0.90 & 0.94 & 0.93 \\
  &   & $x_j,\ j \notin \{2,9\}$          & 0.02 & 0.01 & 0.00 \\ \cmidrule(l){2-6}
  & \multirow{3}{*}{$\xi$}
      & $x_6$                              & 0.96 & 0.97 & 0.97 \\
  &   & $x_9$                              & 0.94 & 0.98 & 0.98 \\
  &   & $x_j,\ j \notin \{6,9\}$          & 0.20 & 0.10 & 0.05 \\ \cmidrule(l){2-6}
  & \multirow{3}{*}{$\kappa$}
      & $x_1$                              & 0.94 & 0.94 & 0.92 \\
  &   & $x_9$                              & 0.95 & 0.95 & 0.94 \\
  &   & $x_j,\ j \notin \{1,9\}$          & 0.01 & 0.01 & 0.01 \\ \midrule
\multirow{6}{*}{$c$}
  & \multirow{3}{*}{$\theta$}
      & $x_1$                              & 0.84 & 0.89 & 0.97 \\
  &   & $x_9$                              & 0.85 & 0.92 & 0.94 \\
  &   & $x_j,\ j \notin \{1,9\}$          & 0.61 & 0.25 & 0.09 \\ \cmidrule(l){2-6}
  & \multirow{3}{*}{$\delta$}
      & $x_2$                              & 0.80 & 0.80 & 0.87 \\
  &   & $x_3$                              & 0.82 & 0.92 & 0.97 \\
  &   & $x_j,\ j \notin \{2,3\}$          & 0.26 & 0.14 & 0.07 \\ \bottomrule
\end{tabular}
\caption{Median posterior inclusion probability of each nonlinear effect, averaged across simulation replications, for different sample sizes. Entries of the form $x_j$ denote covariates with truly active nonlinear effects for that parameter, whereas entries of the form $x_\ell,\ \ell\notin\{j,k\}$ aggregate covariates with truly inactive nonlinear effects. Values close to one indicate strong posterior support for inclusion, whereas values close to zero indicate strong posterior support for exclusion.}
\label{tab:inclusion_accuracy}
\end{table}

Finally, we close this section with a brief description of the performance measures used in the simulation study. For effect selection, we summarize posterior support for each smooth effect through the median posterior inclusion probability. For each simulation replication and each smooth effect, we draw $100$ samples from the fitted variational distribution $q_{\bs\phi}(\bs\zeta_{d,p,j})$ in \eqref{eq:q_zeta}, transform these draws to the corresponding inclusion probabilities, and compute their median, denoted by $\widetilde{p}_{d,p,j}$. Let $z_{d,p,j} \in \{0,1\}$ denote the true inclusion state, where $z_{d,p,j}=1$ if the effect is truly relevant and $z_{d,p,j}=0$ otherwise. Values of $\widetilde{p}_{d,p,j}$ close to one indicate strong posterior support for inclusion of the nonlinear effect, whereas values close to zero indicate strong posterior support for exclusion. We summarize this median posterior inclusion probability across simulation replications for each sample size.

To assess the finite-sample construction underlying the halo densities, we compare representations of the true and fitted halos on the same simulated covariate sample. This avoids conflating model error with numerical approximation of a continuous set. We use the Jaccard distance between the corresponding index sets. For two sets $A,B \subset \{1,\dots,n\}$, we define
\begin{align*}
  J(A,B) = 1 - \frac{|A \cap B|}{|A \cup B|}.
\end{align*}
We compute $J$ for risk halos and favorable halos, separately for each tail event $e\in\mathcal E$, and summarize its distribution across replications for $n \in \{1000, 2000, 4000\}$. Lower values of $J$ indicate better agreement between the fitted and benchmark halo representations.

For risk halos we additionally quantify how much tail probability is captured by a fitted risk-halo representation. Given a set of indices $A \subset \{1,\dots,n\}$, we define the captured tail mass for event $e$ as
\begin{align*}
  \mathrm{TMC}^{(e)}(A)
  &= \frac{\sum_{i \in A} \widehat{\pi}^{(e)}_i}
         {\sum_{i=1}^n \widehat{\pi}^{(e)}_i}.
\end{align*}
For a risk halo at level $\alpha$, we would ideally have $\mathrm{TMC}^{(e)}(\widehat I^{(e)}_{\mathrm{R},\alpha}) \approx \alpha$. We therefore consider the TMC error
\begin{align}
    \bigl|\mathrm{TMC}^{(e)}(\widehat I^{(e)}_{\mathrm{R},\alpha}) - \alpha\bigr|,
    \label{eq:tmc}
\end{align}
and summarize it over replications for each $n \in \{1000, 2000, 4000\}$. To propagate posterior uncertainty into the halo summaries, we draw $100$ samples from the fitted variational distributions $q_{\bs\phi}(\bs\beta_d)$ and $q_{\bs\phi}(\bs\zeta_{d,p,j})$ in \eqref{eq:q_beta} and \eqref{eq:q_zeta}, recompute the corresponding exceedance probabilities and finite-sample halo representations, and then summarize the resulting posterior distribution of each metric by its posterior mean and by the width of the associated $95\%$ posterior interval.

\subsection{Main numerical findings}\label{mainsim}

Table~\ref{tab:inclusion_accuracy} reports spike-and-slab effect selection performance in terms of the median posterior inclusion probability of each nonlinear effect, averaged across the $100$ simulation replications. For each replication and each smooth effect, the median is computed from $100$ samples drawn from the fitted variational distribution $q_{\bs\phi}(\bs\zeta_{d,p,j})$ in \eqref{eq:q_zeta}. Rows labeled by a single $x_j$ correspond to truly active nonlinear effects for that parameter, whereas the remaining-covariate rows aggregate, for each parameter, the median posterior inclusion probabilities over all covariates that are truly non-relevant. In this parameterization, values close to one indicate strong posterior support for inclusion of the nonlinear effect, while values close to zero indicate strong shrinkage toward exclusion.

Overall, Table~\ref{tab:inclusion_accuracy} shows a clear separation between relevant and non-relevant effects. The marginal parameters are selected most cleanly: active effects usually have median inclusion probabilities near one already at $n=1000$, while inactive effects are close to zero and shrink further as $n$ grows. The weakest active signals are $x_7$ in $\nu_1$ and $x_9$ in $\kappa_1$, which remain appreciably below one even at $n=4000$. Selection for the copula parameters is more demanding, especially at $n=1000$, where inactive effects retain larger inclusion probabilities; this difficulty largely disappears by $n=4000$.

We next evaluate the fitted finite-sample representations underlying the risk and favorable halo densities for $Y_1$, $Y_2$, and the joint event $\{Y_1>u_1,\ Y_2>u_2\}$. Thresholds are empirical $0.95$-quantiles and the risk-halo level is $\alpha=0.9$. For each replication, $100$ posterior draws from \eqref{eq:q_beta} and \eqref{eq:q_zeta} are used to recompute exceedance probabilities, finite-sample halo representations, Jaccard distances, and, for risk halos, the tail-mass error in \eqref{eq:tmc}.

Figure~\ref{fig:halos_jaccard_sim} shows that recovery of the finite-sample representations underlying the halo densities improves steadily with sample size. This holds for both risk and favorable halos and for all three components. The joint risk representation is recovered especially accurately, even at $n=1000$, whereas the joint favorable representation is the most difficult object because it depends on whether local exceedance probabilities exceed their average level. Posterior interval widths also contract with $n$, indicating reduced uncertainty in the fitted halo representations.

\begin{figure}[t]
  \centering
  \includegraphics[width=\linewidth]{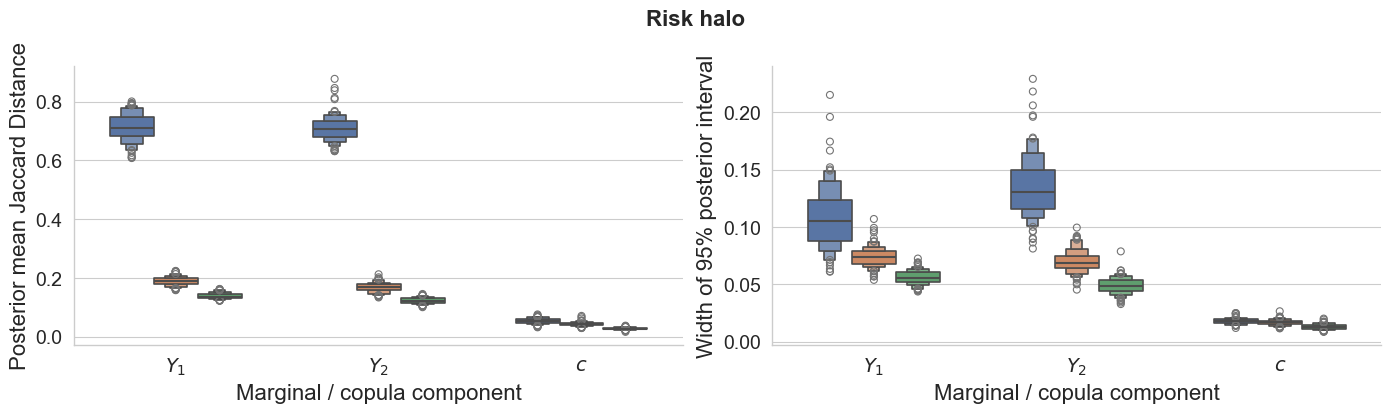}\\
  \includegraphics[width=\linewidth]{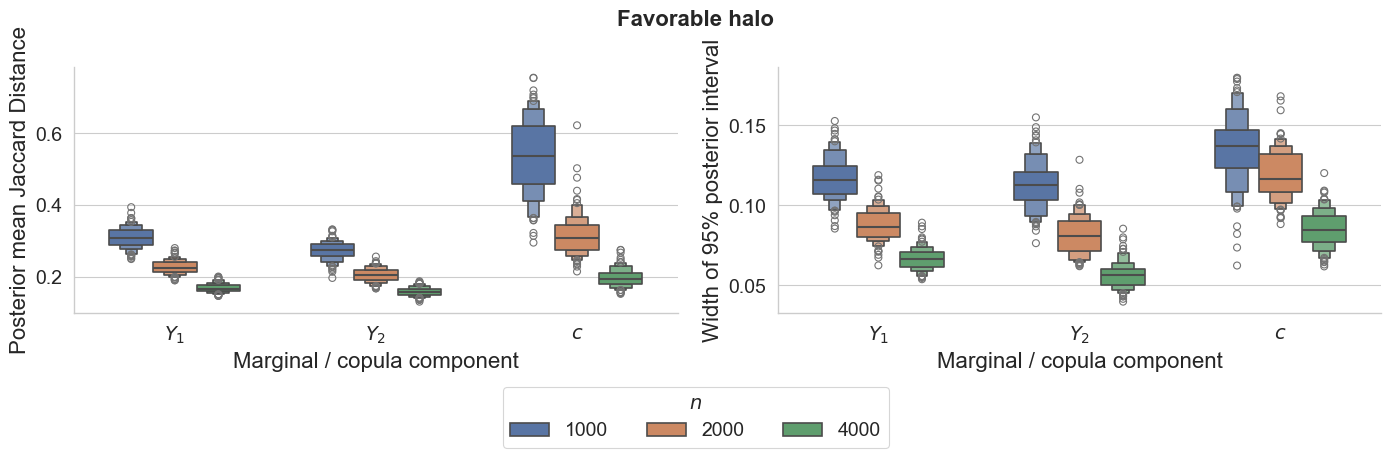}
  \caption{Posterior summaries of the Jaccard distance between benchmark and fitted finite-sample halo representations, for $n \in \{1000,2000,4000\}$, based on $100$ posterior draws per simulation replication. Top row: risk halos. Bottom row: favorable halos. In each row, the left panel shows boxenplots of the posterior mean Jaccard distance across simulation replications, while the right panel shows boxenplots of the width of the corresponding $95\%$ posterior interval.}
  \label{fig:halos_jaccard_sim}
\end{figure}

Figure~\ref{fig:tmc_sim} shows the corresponding tail-mass error for risk halos. The posterior mean error decreases with sample size for both margins and for the joint tail, and the posterior intervals become narrower. Thus the fitted risk-halo representations are accurate not only geometrically, through set overlap, but also probabilistically, through the amount of tail probability they capture.

\begin{figure}[t]
  \centering
  \includegraphics[width=\linewidth]{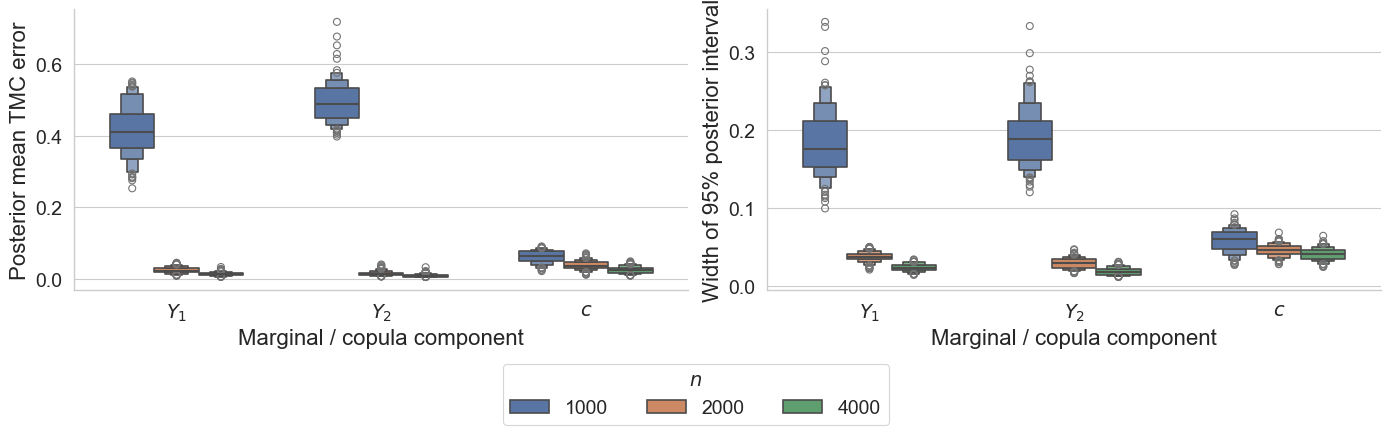}
  \caption{Posterior summaries of the tail-mass error from the target level $\alpha=0.9$, as defined in \eqref{eq:tmc}, for $Y_1$, $Y_2$, and the joint tail event, with $n \in \{1000,2000,4000\}$, based on $100$ posterior draws per simulation replication. The left panel shows boxenplots of the posterior mean TMC error across simulation replications, while the right panel shows boxenplots of the width of the corresponding $95\%$ posterior interval.}
  \label{fig:tmc_sim}
\end{figure}

Taken together, the simulation results show that the proposed framework identifies the active nonlinear covariate effects and translates this structure into increasingly accurate recovery of the halo representations used to estimate halo densities, with posterior uncertainty decreasing as the sample size grows. Appendix~\ref{app:simulation_additional} reports an additional simulation scenario with a different allocation of relevant covariates and effect shapes, leading to the same qualitative conclusions.

\section{Air-pollution extremes in Edinburgh}\label{sec:application}

\subsection{Data description and rationale for the analysis}

We illustrate the proposed methodology using an air-quality dataset collected at an urban monitoring station in Edinburgh, UK. Air pollution is of major public health relevance, with studies documenting its health effects in Scotland and linking it to respiratory hospital admissions \citep{lee2009, huang2022}. The City Council has openly acknowledged in its 2024 \emph{Air Quality Action Plan} the current challenges and the serious public health consequences of inadequate pollutant management \citep{edinburgh2024}:

\begin{quote}\footnotesize 
  ``Improvements in air quality have already been achieved across Edinburgh, which should be celebrated. These improvements have largely been focused on nitrogen dioxide [...] Even with these successes, from a health perspective there is no safe level of certain regulated pollutants, with Particulate Matter (PM$_{2.5}$) now being among those which health experts are most concerned about. [...] Exposure to PM$_{2.5}$ can cause damage to the brain, our cardiovascular and respiratory system, provoking, for example stroke, lung cancer and chronic obstructive pulmonary disease (COPD).''
\end{quote}
The data were gathered from Air Quality in Scotland (\url{http://www.scottishairquality.co.uk}) and consist of hourly measurements from 2025; we conduct a complete-case analysis, yielding a final sample of $n = 8\,081$ data points. 

Our goals are to learn the risk and favorable halo densities associated with extreme pollution events using a fitted multivariate distributional regression model for fine particulate matter (PM$_{2.5}$ in $\mu\text{g}/\text{m}^3$) and nitrogen dioxide (NO$_2$ in $\mu\text{g}/\text{m}^3$) concentrations, and to identify which covariates drive the marginal tails and their dependence. As covariates, we use hour of day, day of the week, month, air temperature, wind direction, and wind speed; %
for hour, day of the week, month, and wind direction, we employ cyclic P-splines to respect their periodic structure, while temperature and wind speed are modeled with standard P-splines. 

\subsection{\texorpdfstring{Marginal and joint halo densities for PM$_{2.5}$ and NO$_2$}{Marginal and joint halo densities for PM2.5 and NO2}}
\subsubsection*{Modeling and inference}
In the halo analysis we focus on high pollution levels, setting marginal thresholds equal to the empirical $0.95$-quantiles of PM$_{2.5}$ and NO$_2$, respectively, and fixing $\alpha = 0.9$ so that the resulting risk halos define the covariate region capturing $90\%$ of pollutant exceedances.
The resulting thresholds are approximately $14.88$ and $33.24\,\mu\text{g}/\text{m}^3$ for PM$_{2.5}$ and NO$_2$, respectively. These represent clearly elevated concentrations within the local hourly distribution and are numerically approximately $3.0$ and $3.3$ times the corresponding WHO guideline values for annual mean exposure \citep{who2021air}.

\begin{figure}
  \centering
  \includegraphics[width=\linewidth]{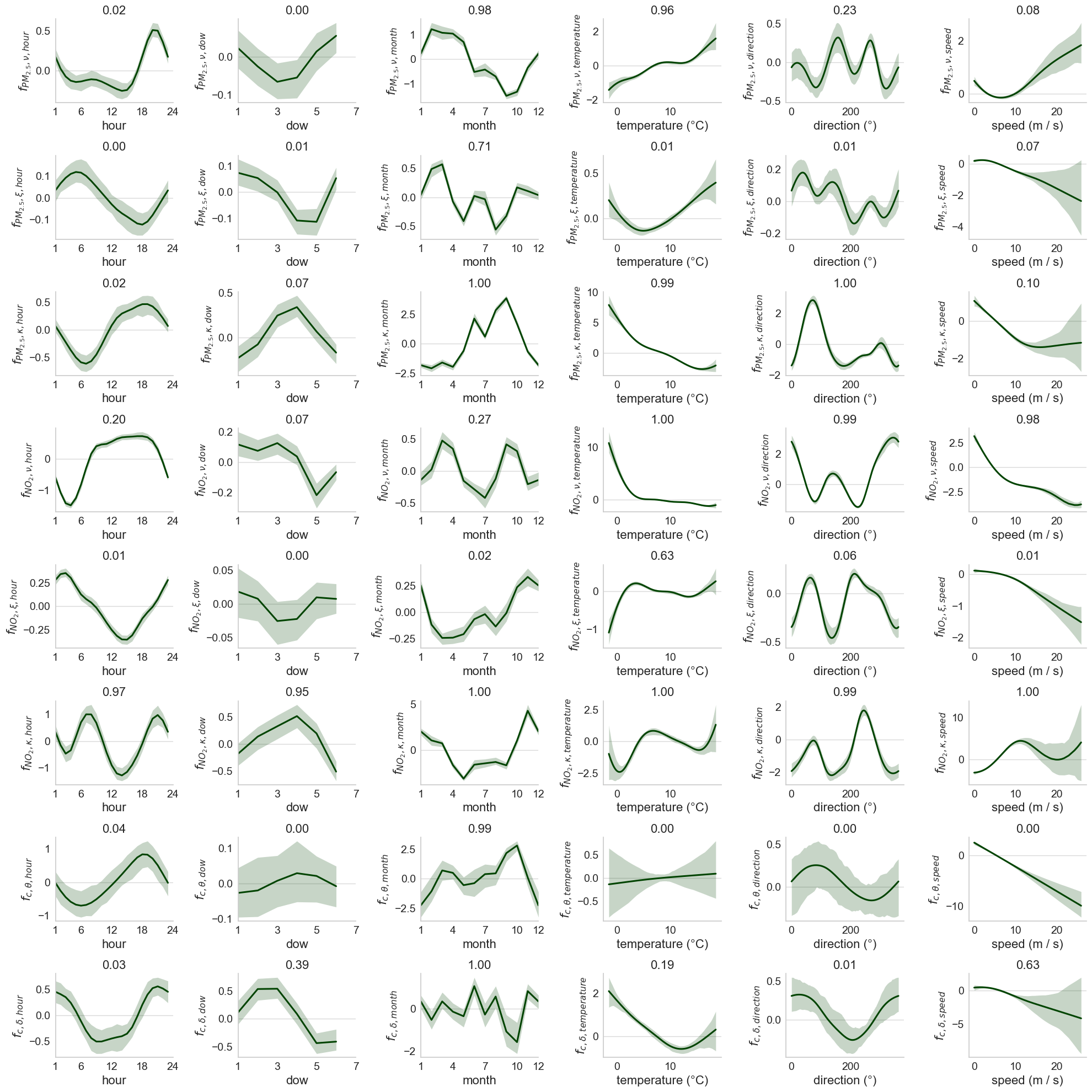}
  \caption{Estimated smooth effects for all parameters and covariates. Rows: $(\nu,\xi,\kappa)$ of PM$_{2.5}$ (first three), $(\nu,\xi,\kappa)$ of NO$_2$ (next three), and BB1 copula parameters $(\theta,\delta)$ (last two). Columns: hour, day of the week, month, temperature, wind direction, wind speed. Titles show the posterior mean inclusion probabilities induced by the spike-and-slab prior for the nonlinear effects.}
  \label{fig:effects_app}
\end{figure}

To learn halo densities from the data, we use the proposed joint spike-and-slab EGPD regression model described in Section~\ref{model2}. Specifically, for each marginal component of the bivariate response, we fit an EGPD regression with parameters $(\nu_d, \xi_d, \kappa_d)$ and couple the marginals using a BB1 copula to model the joint distribution; each parameter is linked to the covariates via a structured additive predictor with P-splines and spike-and-slab effect selection on the nonlinear spline coefficients, using the prior hierarchy introduced in Section~\ref{sec:spike-slab-model}. For model fitting, we use virtually the same settings as in the simulation study in Section~\ref{section:SVIModel}. In particular, we use the same number of interior knots, and the posterior approximation is performed using the same stochastic variational inference scheme as in the simulation study, with analogous prior specifications. The only exception is that for the slab variance we use a narrower inverse-gamma prior than in the simulation study, $\psi^2_{d,p,j} \sim \mathrm{Inverse\text{-}Gamma}(a_{d, p, j}, b_{d, p, j})$ with $a_{d, p, j} = 5$ and $b_{d, p, j} = 10$, whereas the simulations used $\mathrm{Inverse\text{-}Gamma}(a_{d, p, j}, b_{d, p, j})$ with $a_{d, p, j} = 5$ and $b_{d, p, j} = 20$; given the larger effective sample size in this application, the choice $(5,10)$ provides a more concentrated prior on $\psi^2$, which stabilizes the spline amplitudes and helps to retain weaker but practically important effects. %

\subsubsection*{Estimated effects, relevant covariates, and diagnostics}

Figure~\ref{fig:effects_app} summarizes the estimated smooth effects for all parameters and covariates. Each panel shows the posterior mean function and a pointwise $95\%$ credible band, and the panel title reports the estimated inclusion probability of the corresponding nonlinear effect. This figure also summarizes the effect-selection component of the analysis: the spike-and-slab prior determines which nonlinear covariate effects are retained for each marginal and copula parameter, thereby distinguishing nonlinear covariate effects associated with the marginal parameters from those associated with joint dependence.

The first three rows correspond to the parameters $(\nu_1,\xi_1,\kappa_1)$ of the PM$_{2.5}$ marginal, the next three rows to $(\nu_2,\xi_2,\kappa_2)$ of the NO$_2$ marginal, and the last two rows to the BB1 copula parameters $(\theta,\delta)$. A clear difference emerges between the two marginals. For PM$_{2.5}$, the dominant effects are seasonal and meteorological: month and temperature have very high inclusion probabilities for $\nu_1$ and $\kappa_1$, while the nonlinear components for hour of day and day of week receive little posterior support. This suggests that PM$_{2.5}$ extremes are driven mainly by broader seasonal and atmospheric conditions rather than by short-term temporal variation. For NO$_2$, by contrast, the model detects a richer structure. In particular, temperature, wind direction, and wind speed are strongly relevant for $\nu_2$, and the tail parameter $\kappa_2$ shows high inclusion probabilities for nearly all covariates, including hour of day and day of week. This is consistent with NO$_2$ being shaped both by traffic-related temporal cycles and by meteorological dispersion conditions.

The copula effects are more parsimonious. For both $\theta$ and $\delta$, month stands out as the most consistently relevant covariate, whereas most other effects are strongly penalized. Hence, the dependence between extreme PM$_{2.5}$ and NO$_2$ levels appears to vary primarily on a seasonal scale, with much weaker evidence for additional short-term modulation. Thus, the spike-and-slab component not only serves as a regularization device but also provides an interpretable comparison of the nonlinear covariate effects associated with marginal tail behavior and extremal dependence.

Additional posterior checks for the fitted marginal and joint distributions, together with exploratory diagnostics of covariate dependence within the fitted halos, are reported in~\ref{app:application_additional}.

\subsubsection*{Visualizing fitted tail halos}

Based on the fitted regression model, we compute marginal and joint exceedance probabilities at the observed covariate values and form the fitted finite-sample representations described in Section~\ref{sec:halo_construction}. To visualize the resulting risk and favorable halos, we use the membership-weighted summary in \eqref{eq:kde}; pointwise posterior credible bands are obtained from the draw-specific halo densities in \eqref{eq:kde_draw}.

\begin{figure}[t]
  \centering
  \includegraphics[width=0.7\linewidth]{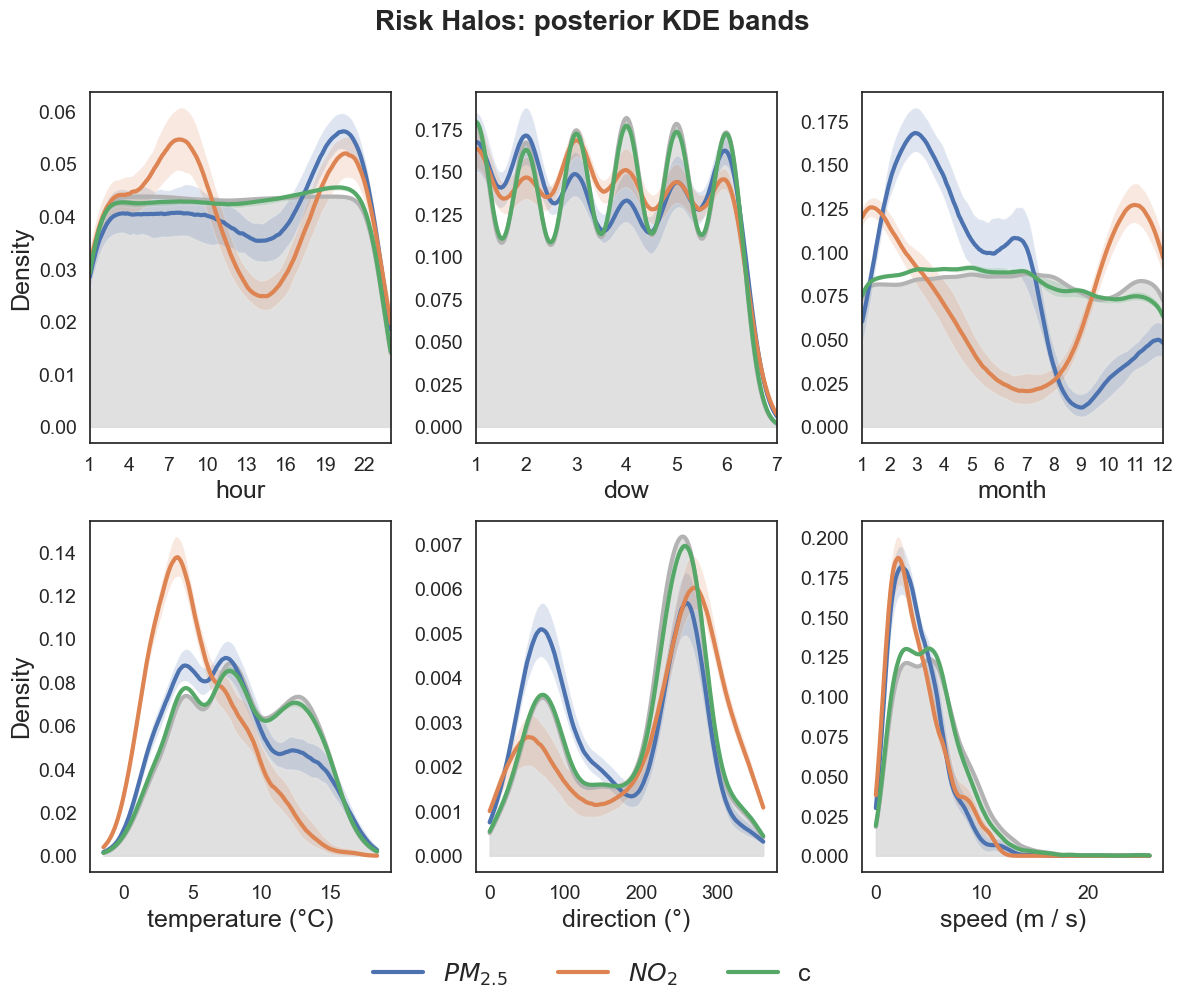}
  \caption{Posterior risk-halo densities for PM$_{2.5}$, NO$_2$, and their joint tail for the halo analysis with $\alpha=0.90$ (colored curves with shaded credible bands), compared with the overall covariate distribution (gray).}
  \label{fig:h_risk_app_90}
\end{figure}

\begin{figure}[t]
  \centering
  \includegraphics[width=0.7\linewidth]{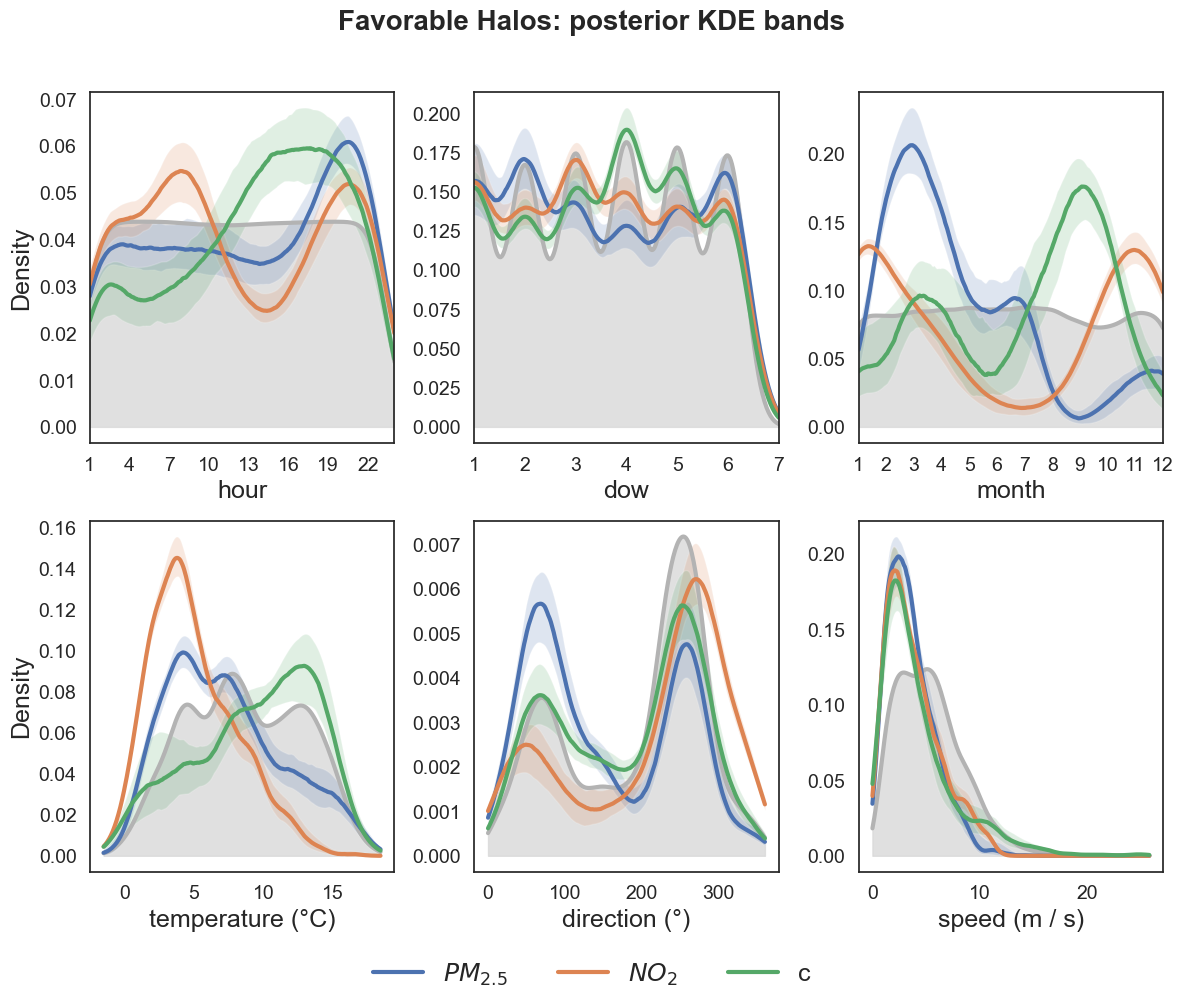}
  \caption{Posterior favorable-halo densities for PM$_{2.5}$, NO$_2$, and their joint tail (colored curves with shaded credible bands), compared with the overall covariate distribution (gray).}
  \label{fig:h_fav_app}
\end{figure}

Figure~\ref{fig:h_risk_app_90} shows the posterior risk-halo densities. For hour of day, the NO$_2$ risk-halo density has pronounced morning and evening peaks, whereas the PM$_{2.5}$ risk-halo density is most concentrated in the evening. Seasonally, NO$_2$ risk is elevated around February and November, while PM$_{2.5}$ is more concentrated in spring. Temperature and wind direction are particularly discriminating for NO$_2$: the risk density is enhanced around $5^\circ\mathrm{C}$ and around westerly winds, approximately $90^\circ$. Wind speed is less distinctive, with risk-halo densities closer to the overall marginal distribution. The joint risk-halo density is also closer to the background covariate distribution than the marginal risk-halo densities, suggesting that joint extremes occur under a broader set of covariate configurations than individual pollutant extremes.

Figure~\ref{fig:h_fav_app} gives the corresponding fitted favorable-halo densities. The joint favorable-halo density is concentrated in the late afternoon and early evening for hour of day, and it shows a visible early-autumn peak for month, highlighting covariate configurations where the joint exceedance probability is elevated relative to its average level. Temperature remains more influential for NO$_2$ than for PM$_{2.5}$, and wind-direction densities remain centered around westerly directions. The posterior bands are wider for several joint favorable-halo densities than for the corresponding risk-halo densities, reflecting the greater sensitivity of favorable halos to posterior fluctuations around the marginal exceedance level.

The risk- and favorable-halo tables in the Supplementary Material, \autoref{tab:risk_halo_examples_application} and \autoref{tab:favourable_halo_examples_application}, report representative observations with the highest posterior halo-membership probabilities. They make the role of temperature particularly clear. The NO$_2$ risk and favorable examples are concentrated under cold conditions, with temperatures close to zero in the strongest case and mostly westerly winds. The PM$_{2.5}$ examples are mostly located in spring months and under mild-to-cool temperatures, while the joint favorable examples occur in spring or autumn and in afternoon or evening hours. These observations are illustrative examples of the fitted halo regions rather than additional inferential targets.

Overall, the application illustrates the distinction between covariate configurations that carry most of the tail mass and configurations where tail risk is elevated relative to the background. This distinction is especially visible for the joint tail, where the risk halo is comparatively stable but the favorable halo is more localized and more uncertain.

\section{Discussion}\label{sec:discussion}

This paper introduced tail halos as covariate-defined regions associated with extreme events. Risk halos identify the smallest regions that account for a chosen fraction of the tail mass, while favorable halos identify regions where the conditional exceedance probability is larger than its marginal level. Halo densities are a key part of the framework: they turn a set-valued tail object into distributional summaries of the covariates inside the halo. This is what makes the method usable beyond very low-dimensional covariate spaces. The simulations show that the proposed Bayesian structured additive distributional regression model recovers relevant nonlinear covariate effects and yields increasingly accurate finite-sample halo representations as the sample size grows. The Edinburgh air-pollution application illustrates how the approach separates covariate configurations that explain most exceedances from those where exceedance risk is elevated relative to the background.

The proposal connects extreme-value regression with ideas from random sets and set-valued inference \citep{molchanov2005}. This connection is useful because the scientific target is not only a parameter or a fitted smooth function, but a region of covariate space and the probability law induced by that region. The distinction matters in interpretation. A regression effect can indicate that a covariate changes a tail parameter, whereas a halo density shows how that covariate is represented among the configurations that carry tail mass or elevated tail risk.

Several methodological questions remain. One is how best to summarize halo laws when the covariates themselves are structured objects, such as curves, spatial fields, images, or networks. In such settings, one-dimensional marginal densities may be replaced by functional summaries, spatial contrasts, graph features, or projections chosen for scientific interpretability. A second direction is uncertainty quantification for features of the halo itself, such as volume, connected components, boundary stability, or contrasts between risk and favorable halos. These questions are close in spirit to statistical problems for random closed sets, but with the additional structure that the set is generated by a tail probability.

Another natural extension is to make the inferential model more directly halo-aware. In the present formulation, nonlinear-effect selection acts on the penalized spline components of the structured additive predictors, and the halos are derived from the resulting fitted tail probabilities. Future priors or loss functions could instead target the stability, sparsity, or interpretability of the halo densities themselves. This would bring the statistical model even closer to the scientific questions that motivate tail halos: where tail events concentrate, where their probability is unusually high, and how those regions should be summarized when the covariate space is too rich to inspect directly.

\bibliographystyle{agsm}
\bibliography{bibliography}

\appendix
\section{Appendix}

\subsection{Proof of the risk-halo characterization}
\label{app:pl_convergence_risk_halo}

\begin{proof}[Proof of Proposition~\ref{prop:risk_halo_characterization}]
Let $p_u$, $\pi_u$, and $Q_u$ be as defined in \eqref{mu_mar}--\eqref{eq:tail_law}. Let $\X_u\sim Q_u$ and $Z_u=\pi_u(\X_u)$. Suppose that the law of $Z_u$ is atom-free. For $\alpha\in(0,1)$, set
\begin{equation}\label{eq:proof_threshold_set}
  \tau=\tau_{u,\alpha}=F_{Z_u}^{-1}(1-\alpha),
  \qquad
  H=\{\x:\pi_u(\x)>\tau\}.
\end{equation}
Because $F_{Z_u}$ is continuous, the quantile in \eqref{eq:proof_threshold_set} satisfies
\begin{equation}\label{eq:proof_exact_tail_mass}
  Q_u(H)=Q_u(Z_u>\tau)=\alpha.
\end{equation}
Moreover, $Z_u>0$ holds $Q_u$-almost surely, since
\begin{equation}\label{eq:proof_positive_score}
  Q_u(Z_u=0)
  =\frac{1}{p_u}\int_{\{\pi_u=0\}}\pi_u\,\mathrm{d}P_{\X}=0.
\end{equation}
Thus $F_{Z_u}(0)=0$, and because $1-\alpha>0$, \eqref{eq:proof_threshold_set} implies that $\tau>0$.

Let $A$ be any feasible set in \eqref{eq:risk_halo_optimization}, so that $Q_u(A)\geq\alpha=Q_u(H)$ by \eqref{eq:proof_exact_tail_mass}. Define
\begin{equation}\label{eq:proof_nonnegative_integrand}
  I_A
  :=\{\mathbf{1}_H-\mathbf{1}_A\}(\pi_u-\tau)
  \geq0.
\end{equation}
Indeed, on $H\setminus A$ both factors in \eqref{eq:proof_nonnegative_integrand} are positive, on $A\setminus H$ both are nonpositive, and elsewhere their product is zero. Consequently,
\begin{equation}\label{eq:proof_optimality_chain}
  0\leq \int I_A\,\mathrm{d}P_{\X}
  =p_u\{Q_u(H)-Q_u(A)\}-\tau\{P_{\X}(H)-P_{\X}(A)\}
  \leq-\tau\{P_{\X}(H)-P_{\X}(A)\}.
\end{equation}
The final inequality in \eqref{eq:proof_optimality_chain} uses $Q_u(A)\geq Q_u(H)$. Since $\tau>0$, \eqref{eq:proof_optimality_chain} gives $P_{\X}(H)\leq P_{\X}(A)$, proving optimality.

Now suppose that $A$ is also a minimizer. Then $P_{\X}(A)=P_{\X}(H)$. Substituting this equality and \eqref{eq:proof_exact_tail_mass} into the identity in \eqref{eq:proof_optimality_chain} gives
\[
  0
  \leq \int I_A\,\mathrm{d}P_{\X}
  =p_u\{\alpha-Q_u(A)\}
  \leq0,
\]
where the last inequality follows from the feasibility of $A$. Hence $Q_u(A)=\alpha$ and $\int I_A\,\mathrm{d}P_{\X}=0$. Since $I_A\geq0$ by \eqref{eq:proof_nonnegative_integrand}, it follows that $I_A=0$ $P_{\X}$-almost everywhere. By the strict sign of $\pi_u-\tau$ off the boundary level set, the symmetric difference between $A$ and $H$ can therefore lie only in $\{\pi_u=\tau\}$, up to a $P_{\X}$-null set. Atom-freeness gives
\[
  0=Q_u(\pi_u=\tau)
  =\frac{\tau}{p_u}P_{\X}(\pi_u=\tau).
\]
The positivity of $\tau$, established using \eqref{eq:proof_positive_score}, therefore implies $P_{\X}(\pi_u=\tau)=0$. Hence $P_{\X}(A\mathbin{\triangle}H)=0$, proving uniqueness up to null sets. Finally, if $\alpha_1<\alpha_2$, then $\tau_{u,\alpha_1}\geq\tau_{u,\alpha_2}$; the definition of $H$ in \eqref{eq:proof_threshold_set} then proves nestedness.

For completeness, suppose now that $\mathcal{X}$ is closed and that $Q_u$ has full support on $\mathcal{X}$. Let $\alpha_m\uparrow1$ and write $H_m=\mathscr H_{u,\alpha_m}$. Since $H_m\subseteq\mathcal{X}$ and $\mathcal{X}$ is closed, every outer limit point lies in $\mathcal{X}$. Conversely, fix $\x\in\mathcal{X}$ and any relatively open neighborhood $G$ of $\x$. Full support gives $Q_u(G)>0$. If $G\cap H_m=\varnothing$ for infinitely many $m$, then along that subsequence
\begin{equation}\label{eq:proof_pk_bound}
  \alpha_m
  =Q_u(H_m)
  \leq Q_u(\mathcal{X}\setminus G)
  =1-Q_u(G)
  <1,
\end{equation}
where the equality uses \eqref{eq:proof_exact_tail_mass}. The fixed upper bound in \eqref{eq:proof_pk_bound} contradicts $\alpha_m\uparrow1$. Thus every neighborhood of every point of $\mathcal{X}$ eventually intersects $H_m$. The neighborhood characterization of Painlev\'e--Kuratowski convergence \citep[see, e.g.,][]{molchanov2005} now yields
\[
  \mathscr H_{u,\alpha}
  \xrightarrow[\mathrm{PK}]{}
  \mathcal{X}
  \qquad\text{as }\alpha\uparrow1.
\]
\end{proof}

\subsection{Spline constraints and penalty reparameterization}\label{sec:spline_constraints}

This section records the implementation details used to make the additive predictors identifiable and to express each spline effect in the penalized/unpenalized form used by the Gaussian priors. These steps are straightforward but included here for completeness.

For a smooth effect $f_{i,d,p,j}=\X_{d,p,j}\bs\beta_{d,p,j}$, identifiability is enforced by centering the fitted effect over the observed covariates,
\begin{align*}
\sum_{i=1}^n f_{i,d,p,j}^{\rm new}(x_{i,d,p,j},\bs\beta_{d,p,j})=0.
\end{align*}
Let $\C_{d,p,j}$ be the vector of column means of $\X_{d,p,j}$. A QR decomposition of $\C_{d,p,j}$ gives an orthogonal basis $\Q_{d,p,j}=[\mathbf Z^{(a)}_{d,p,j},\mathbf Z^{(b)}_{d,p,j}]$, where $\mathbf Z^{(b)}_{d,p,j}$ spans the subspace satisfying the centering constraint. We therefore use the constrained design and penalty
\begin{align*}
\X_{d,p,j}^{\rm new}=\X_{d,p,j}\mathbf Z^{(b)}_{d,p,j}, \qquad \K_{d,p,j}^{\rm new} =\left(\mathbf Z^{(b)}_{d,p,j}\right)^\T \K_{d,p,j}\mathbf Z^{(b)}_{d,p,j}.
\end{align*}
For one distributional parameter $p$ in component $d$, the full design and penalty matrices are then
\begin{align}
\X_{d,p} = \left[\mathbf 1,\X_{d,p,1}^{\rm new},\dots,\X_{d,p,J_{d,p}}^{\rm new}\right]
\label{eq:X_p}
\end{align}
\begin{align*}
\K_{d,p}(\bs\tau_{d,p}) = \operatorname{blockdiag}\left\{0,\tau^{-2}_{d,p,1}\K^{\rm new}_{d,p,1},\dots, \tau^{-2}_{d,p,J_{d,p}}\K^{\rm new}_{d,p,J_{d,p}}\right\}.
\end{align*}
To separate penalized and unpenalized components, let
\begin{align*}
\K^{\rm new}_{d,p,j} = 
\U_{d,p,j}
\begin{pmatrix}
\boldsymbol\Lambda_{d,p,j}^{+} & 0 \\
0 & 0
\end{pmatrix}
\U_{d,p,j}^\T,
\end{align*}
where $\boldsymbol\Lambda_{d,p,j}^{+} = \operatorname{diag}(\lambda_1,\dots,\lambda_r)$ contains the positive eigenvalues of $\K^{\rm new}_{d,p,j}$. Writing $\U_{d,p,j}=[\U^{+}_{d,p,j},\U^{0}_{d,p,j}]$, with $\U^{0}_{d,p,j}$ spanning the null space of the penalty, we define
\begin{align*}
\X_{d,p,j}^{\rm penalized} = \X_{d,p,j}^{\rm new} \U^{+}_{d,p,j} \left(\boldsymbol\Lambda_{d,p,j}^{+}\right)^{-1/2}, \qquad \X_{d,p,j}^{\rm unpenalized} = \X_{d,p,j}^{\rm new} \U^{0}_{d,p,j}.
\end{align*}
Thus
\begin{align*}
\X_{d,p,j}^{\rm transformed} = \left[\X_{d,p,j}^{\rm penalized}, \X_{d,p,j}^{\rm unpenalized} \right].
\end{align*}
Under this reparameterization, the penalty becomes the identity on the penalized block and zero on the unpenalized block, giving the usual mixed-model representation of spline effects used in the structured additive predictor.

\subsection{Additional simulation scenario}\label{app:simulation_additional}

This section complements the simulation study in the main paper by changing the allocation of relevant covariates and effect shapes across the marginal and copula parameters. The aim is to check that the main conclusions are not tied to the particular regression functions used in the primary experiment. The overall simulation design is the same as in the main text: we consider $n \in \{1000, 2000, 4000\}$, generate $100$ independent datasets for each sample size, and fit the same multivariate distributional regression model with the same spike-and-slab prior specification, hyperparameters, and SVI settings.

The only change is the assignment of relevant covariates and nonlinear effect functions to the distributional parameters. In this scenario the true parameter functions are given by
\begin{equation*}
\begin{aligned}
\nu_{1}(\x)
&= \operatorname{softplus}\bigl(0.77 - \sin x_{1} - 1.2\cos x_{4}\bigr), \\
\xi_{1}(\x)
&= \operatorname{softplus}\bigl(-0.44 + 0.7\cos x_{5} - 1.2\cos x_{6}\bigr) - 0.5, \\
\kappa_{1}(\x)
&= \operatorname{softplus}\bigl(0.17 + 0.7\cos x_{5} + 0.8\sin x_{7}\bigr), \\[0.4em]
\nu_{2}(\x)
&= \operatorname{softplus}\bigl(0.93 - 0.8\sin x_{4} + 0.9\sin x_{8}\bigr), \\
\xi_{2}(\x)
&= \operatorname{softplus}\bigl(-0.15 + \cos x_{7} - 0.8\sin x_{9}\bigr) - 0.5, \\
\kappa_{2}(\x)
&= \operatorname{softplus}\bigl(0.42 + 0.8\sin x_{3} + 0.9\sin x_{4}\bigr), \\[0.4em]
\theta(\x)
&= \operatorname{softplus}\bigl(-0.32 + 0.8\sin x_{3} + \cos x_{5}\bigr), \\
\delta(\x)
&= \operatorname{softplus}\bigl(-0.77 + 0.7\cos x_{1} - 0.8\sin x_{2}\bigr) + 1.
\end{aligned}
\label{eq:true-regression-functions-app}
\end{equation*}

\subsubsection{Main numerical findings}

As in the main simulation study, we first evaluate spike-and-slab effect selection and then study the finite-sample halo representations used to construct the risk- and favorable-halo densities.

Table~\ref{tab:inclusion_accuracy_app} reports the median posterior inclusion probability of each nonlinear effect, averaged across the $100$ simulation replications. For each replication and each smooth effect, the median is computed from $100$ samples drawn from the fitted variational distribution of the corresponding spike-and-slab inclusion parameter. Rows labeled by a single $x_j$ correspond to truly active nonlinear effects for that parameter, whereas the remaining-covariate rows aggregate, for each parameter, the median posterior inclusion probabilities over all covariates that are truly non-relevant. In this parameterization, values close to one indicate strong posterior support for inclusion of the nonlinear effect, while values close to zero indicate strong shrinkage toward exclusion.

\begin{table}
\centering
\begin{tabular}{@{}cccccc@{}}
\toprule
\textbf{Response} & \textbf{Parameter} & \textbf{Covariate}
& $n = 1000$ & $n = 2000$ & $n = 4000$ \\ \midrule
\multirow{9}{*}{$y_1$}
  & \multirow{3}{*}{$\nu$}
      & $x_1$                              & 0.98 & 0.98 & 0.98 \\
  &   & $x_4$                              & 0.98 & 0.98 & 0.98 \\
  &   & $x_j,\ j \notin \{1,4\}$          & 0.03 & 0.00 & 0.00 \\ \cmidrule(l){2-6}
  & \multirow{3}{*}{$\xi$}
      & $x_5$                              & 0.72 & 0.76 & 0.79 \\
  &   & $x_6$                              & 0.98 & 0.98 & 0.98 \\
  &   & $x_j,\ j \notin \{5,6\}$          & 0.13 & 0.06 & 0.02 \\ \cmidrule(l){2-6}
  & \multirow{3}{*}{$\kappa$}
      & $x_5$                              & 0.83 & 0.82 & 0.80 \\
  &   & $x_7$                              & 0.96 & 0.96 & 0.95 \\
  &   & $x_j,\ j \notin \{5,7\}$          & 0.02 & 0.01 & 0.01 \\ \midrule
\multirow{9}{*}{$y_2$}
  & \multirow{3}{*}{$\nu$}
      & $x_4$                              & 0.94 & 0.96 & 0.96 \\
  &   & $x_8$                              & 0.98 & 0.98 & 0.97 \\
  &   & $x_j,\ j \notin \{4,8\}$          & 0.03 & 0.01 & 0.00 \\ \cmidrule(l){2-6}
  & \multirow{3}{*}{$\xi$}
      & $x_7$                              & 0.94 & 0.97 & 0.97 \\
  &   & $x_9$                              & 0.99 & 0.99 & 0.99 \\
  &   & $x_j,\ j \notin \{7,9\}$          & 0.09 & 0.05 & 0.02 \\ \cmidrule(l){2-6}
  & \multirow{3}{*}{$\kappa$}
      & $x_3$                              & 0.98 & 0.98 & 0.98 \\
  &   & $x_4$                              & 0.78 & 0.80 & 0.83 \\
  &   & $x_j,\ j \notin \{3,4\}$          & 0.03 & 0.01 & 0.01 \\ \midrule
\multirow{6}{*}{$c$}
  & \multirow{3}{*}{$\theta$}
      & $x_3$                              & 0.76 & 0.85 & 0.83 \\
  &   & $x_5$                              & 0.89 & 0.97 & 0.97 \\
  &   & $x_j,\ j \notin \{3,5\}$          & 0.43 & 0.19 & 0.06 \\ \cmidrule(l){2-6}
  & \multirow{3}{*}{$\delta$}
      & $x_1$                              & 0.75 & 0.88 & 0.96 \\
  &   & $x_2$                              & 0.74 & 0.79 & 0.87 \\
  &   & $x_j,\ j \notin \{1,2\}$          & 0.27 & 0.15 & 0.08 \\ \bottomrule
\end{tabular}
\caption{Median posterior inclusion probability of each nonlinear effect, averaged across simulation replications, for different sample sizes in the additional simulation scenario. Entries of the form $x_j$ denote covariates with truly active nonlinear effects for that parameter, whereas entries of the form $x_\ell,\ \ell\notin\{j,k\}$ aggregate covariates with truly inactive nonlinear effects. Values close to one indicate strong posterior support for inclusion, whereas values close to zero indicate strong posterior support for exclusion.}
\label{tab:inclusion_accuracy_app}
\end{table}

The qualitative message is very similar to that in the main text, but the median representation makes the separation between active and inactive effects even more transparent. Across most marginal parameters, the truly active nonlinear effects have median posterior inclusion probabilities already close to one at $n=1000$, while the truly inactive nonlinear effects have medians close to zero and are shrunk more strongly as the sample size increases. The clearest marginal signals occur in $\nu_1$, $\nu_2$, $\xi_2$ and several components of $\kappa_2$, whereas the weakest relevant effects are found in $\xi_1$ and $\kappa_1$, where one of the two active covariates has a median posterior inclusion probability only in the range $0.7$--$0.8$. As in the main simulation study, the copula parameters remain the most difficult part of the model. In particular, for $\theta$ and $\delta$ the active nonlinear effects are clearly identified, but the non-relevant covariates retain noticeably larger medians at $n=1000$ than in the marginal blocks, indicating weaker separation at the smallest sample size. This distinction fades as $n$ increases, with the medians for non-relevant copula effects moving much closer to zero by $n=4000$. Overall, the additional scenario confirms the same pattern as in the main simulation study: marginal nonlinear effects are easier to identify than dependence effects, and the posterior separation between active and inactive nonlinear effects sharpens substantially as the sample size grows.

As in the main simulation scenario, we form finite-sample representations of the risk and favorable halos for the two margins and for the joint tail event $\{Y_1>u_1,\ Y_2>u_2\}$, using the model-based exceedance probabilities. These representations are the objects from which the corresponding halo densities are constructed. Thresholds $u_1$ and $u_2$ are set to the empirical $0.95$-quantiles of $Y_1$ and $Y_2$, and we fix $\alpha = 0.9$ for the risk halos. To incorporate posterior uncertainty, for each simulated dataset we draw $100$ posterior samples from the fitted variational approximation to the regression coefficients and inclusion parameters, and recompute the corresponding halo representations.

Figure~\ref{fig:halos_app_jaccard_appendix} reports posterior summaries of the Jaccard distance between benchmark and fitted finite-sample halo representations, based on $100$ posterior draws per simulation replication. As in the main text, the left panels show boxenplots of the posterior mean Jaccard distance across replications, while the right panels show the width of the corresponding $95\%$ posterior intervals. For the two margins, the same pattern as in the main text reappears clearly: the posterior mean Jaccard distance decreases steadily as $n$ grows, both for risk and favorable representations. The posterior interval widths also contract with increasing sample size, indicating that uncertainty in the representations underlying the halo-density estimates decreases as more data become available. For the joint component, the contrast between risk and favorable representations is again particularly marked. The joint risk representation is recovered very well even at $n=1000$, with small posterior mean Jaccard distances and narrow posterior intervals. In contrast, the joint favorable representation remains much more difficult to estimate at small $n$, with clearly larger posterior means and wider posterior intervals, although both improve substantially as the sample size increases.

\begin{figure}[t]
  \centering
  \includegraphics[width=\linewidth]{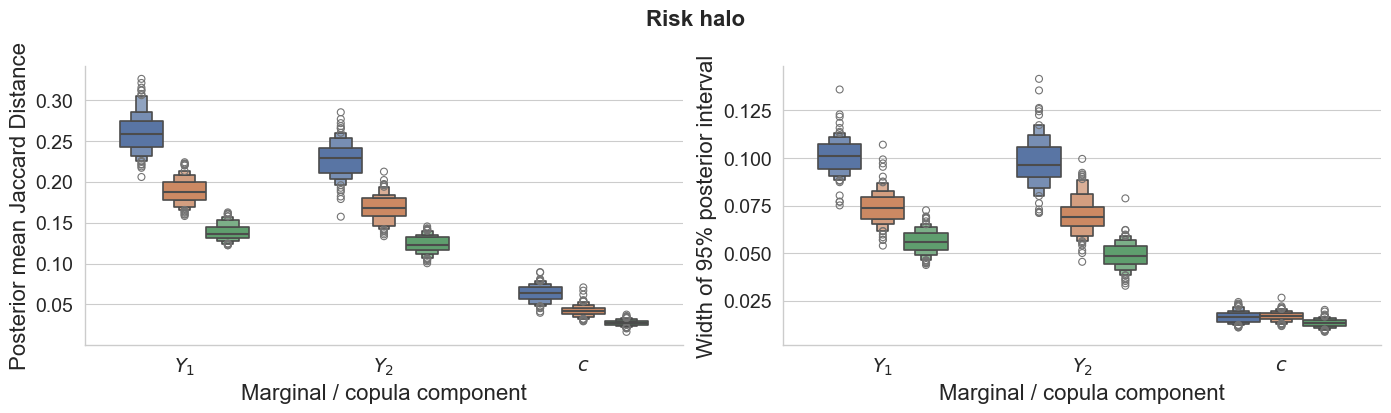}\\
  \includegraphics[width=\linewidth]{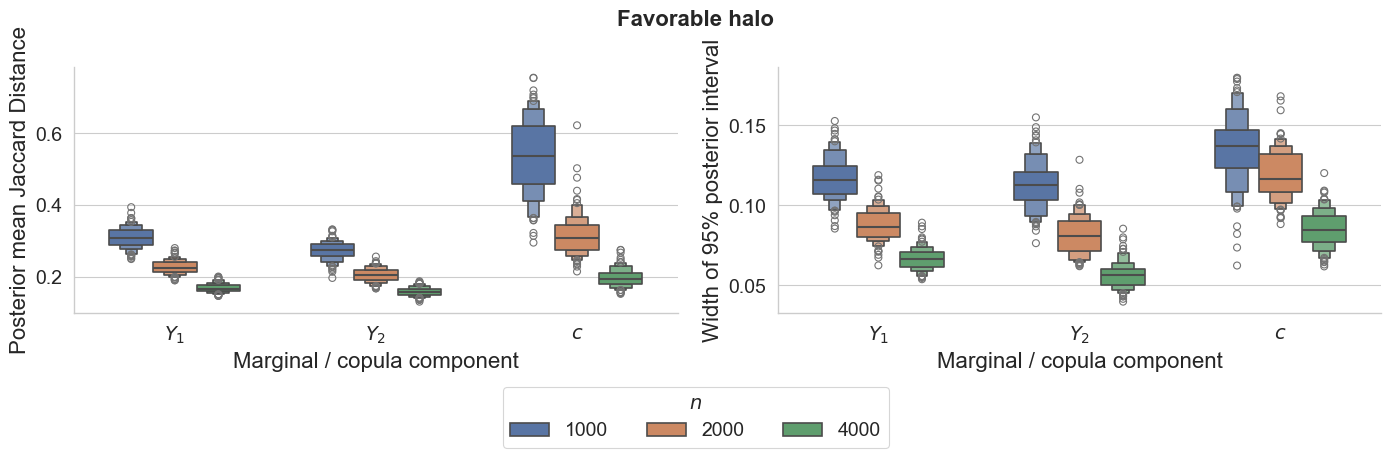}
  \caption{Posterior summaries of the Jaccard distance between benchmark and fitted finite-sample halo representations, for $n \in \{1000,2000,4000\}$, based on $100$ posterior draws per simulation replication. Top row: risk halos. Bottom row: favorable halos. In each row, the left panel shows boxenplots of the posterior mean Jaccard distance across replications, while the right panel shows boxenplots of the width of the corresponding $95\%$ posterior interval.}
  \label{fig:halos_app_jaccard_appendix}
\end{figure}

Figure~\ref{fig:tmc_app_appendix} complements the set-based comparison by showing posterior summaries of the deviation $|\mathrm{TMC}-\alpha|$ for the risk-halo representations. The left panel reports the posterior mean TMC error across replications, while the right panel reports the width of the corresponding $95\%$ posterior interval. For the marginal risk representations, the captured tail mass is already reasonably close to the target level at $n=1000$ and becomes increasingly accurate as $n$ grows. For the joint risk representation, the posterior mean error is larger than for the marginals, but it also decreases steadily with sample size. The same contraction is visible in the posterior interval widths. Hence, even in this alternative scenario, the proposed model captures the correct amount of tail probability with increasing precision and reduced posterior uncertainty, supporting the halo-density summaries built from these representations.
 
\begin{figure}[t]
  \centering
  \includegraphics[width=\linewidth]{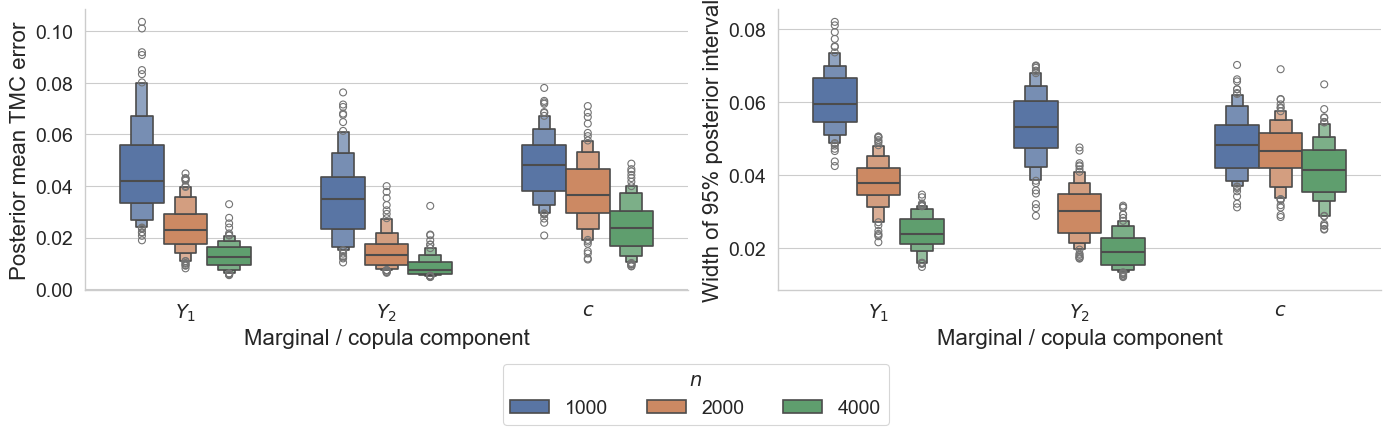}
  \caption{Posterior summaries of the tail-mass error from the target level $\alpha=0.9$, for $Y_1$, $Y_2$, and the joint tail event, with $n \in \{1000,2000,4000\}$, based on $100$ posterior draws per simulation replication. The left panel shows boxenplots of the posterior mean TMC error across replications, while the right panel shows boxenplots of the width of the corresponding $95\%$ posterior interval.}
  \label{fig:tmc_app_appendix}
\end{figure}

Overall, this additional simulation confirms the conclusions of the main study. The spike-and-slab mechanism is able to identify the active nonlinear covariate effects with high posterior support, while the fitted model yields increasingly accurate finite-sample representations for the risk- and favorable-halo densities as the sample size grows. In particular, the distinction between the relatively stable joint risk representation and the more demanding joint favorable representation is reproduced clearly in this alternative configuration, now together with a clear reduction in posterior uncertainty as $n$ increases.

\subsection{Additional diagnostics for the Edinburgh application}\label{app:application_additional}

\subsubsection{Marginal Q--Q diagnostics}

To assess the adequacy of the marginal EGPD fits, we use probability integral transform (PIT) Q--Q plots. For each marginal $Y_d$ we compute
\begin{align*}
\epsilon_{i,d}=\Phi^{-1}\left(F_d\bigl(y_{i,d}\mid\widehat{\nu}_d(\mathbf{x}_i),\widehat{\xi}_d(\mathbf{x}_i),\widehat{\kappa}_d(\mathbf{x}_i)\bigr)\right), \qquad i=1,\dots,n,
\end{align*}
and compare the ordered Gaussianized PIT values to the quantiles of a $\mathrm{Normal}(0,1)$ distribution. Figure~\ref{fig:qq_app_90} shows the resulting Q--Q plots for the two margins, together with pointwise posterior bands obtained by sampling $100$ times from the variational approximation and recomputing the transformed quantiles. For both PM$_{2.5}$ and NO$_2$, the empirical quantiles follow the theoretical line very closely over the central part of the distribution, indicating that the fitted EGPD marginals provide a good overall calibration. The posterior bands remain narrow over most of the range and widen only in the most extreme quantiles, reflecting the greater posterior uncertainty in the tails. Mild departures from the diagonal are visible in the far lower tail, particularly for NO$_2$, where a few extreme observations induce a more pronounced deviation. In the upper tail, which is the most relevant region for the halo analysis, the agreement with the theoretical line remains good for both marginals, with only moderate uncertainty at the very largest quantiles. Overall, these diagnostics indicate that the EGPD marginals provide a satisfactory fit, especially in the upper tail where the halo-density summaries are most sensitive.

\begin{figure}[t]
  \centering
  \includegraphics[width=0.7\linewidth]{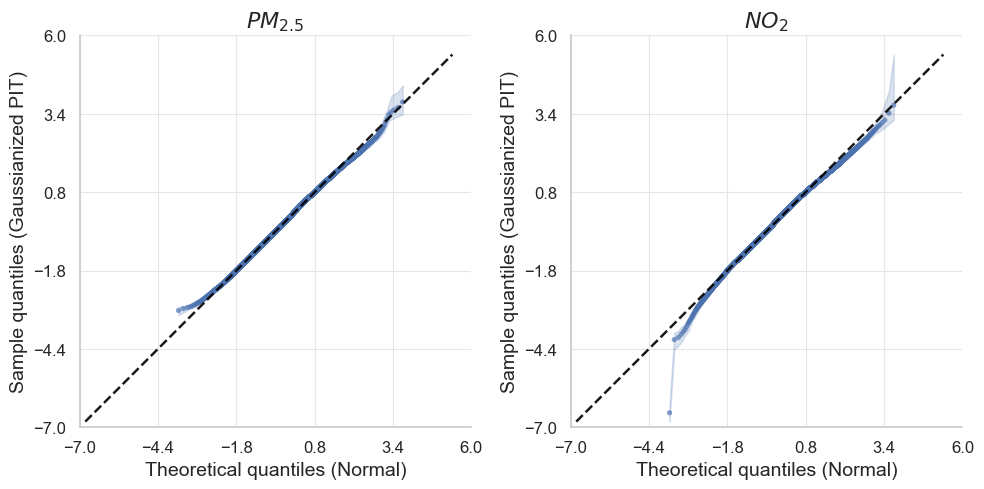}
  \caption{Q--Q plots of the Gaussianized PIT values for the PM$_{2.5}$ (left) and NO$_2$ (right) marginals for the halo analysis with $\alpha=0.90$, with pointwise posterior bands obtained from the variational approximation. The near-linear behavior along the diagonal indicates an adequate marginal fit, with mild deviations confined to the most extreme quantiles.}
  \label{fig:qq_app_90}
\end{figure}

\subsubsection{Posterior inclusion probabilities in the Edinburgh application}

Figure~\ref{fig:effects_rho_app} displays the variational posterior distributions of the inclusion probabilities for all nonlinear effects in the Edinburgh application. Effects with posterior meadian close to $0$ or $1$ generally show distributions strongly concentrated near the corresponding boundary, indicating clear evidence for exclusion or inclusion. Intermediate inclusion probabilities are associated with more dispersed or skewed posterior distributions, reflecting greater uncertainty about whether the corresponding nonlinear effect should be active. The strongest and most stable inclusion signals occur for month and temperature in the PM$_{2.5}$ marginal, for temperature, wind direction, and wind speed in the NO$_2$ marginal, and for month in the copula parameters.

\begin{figure}[t]
  \centering
  \includegraphics[width=\linewidth]{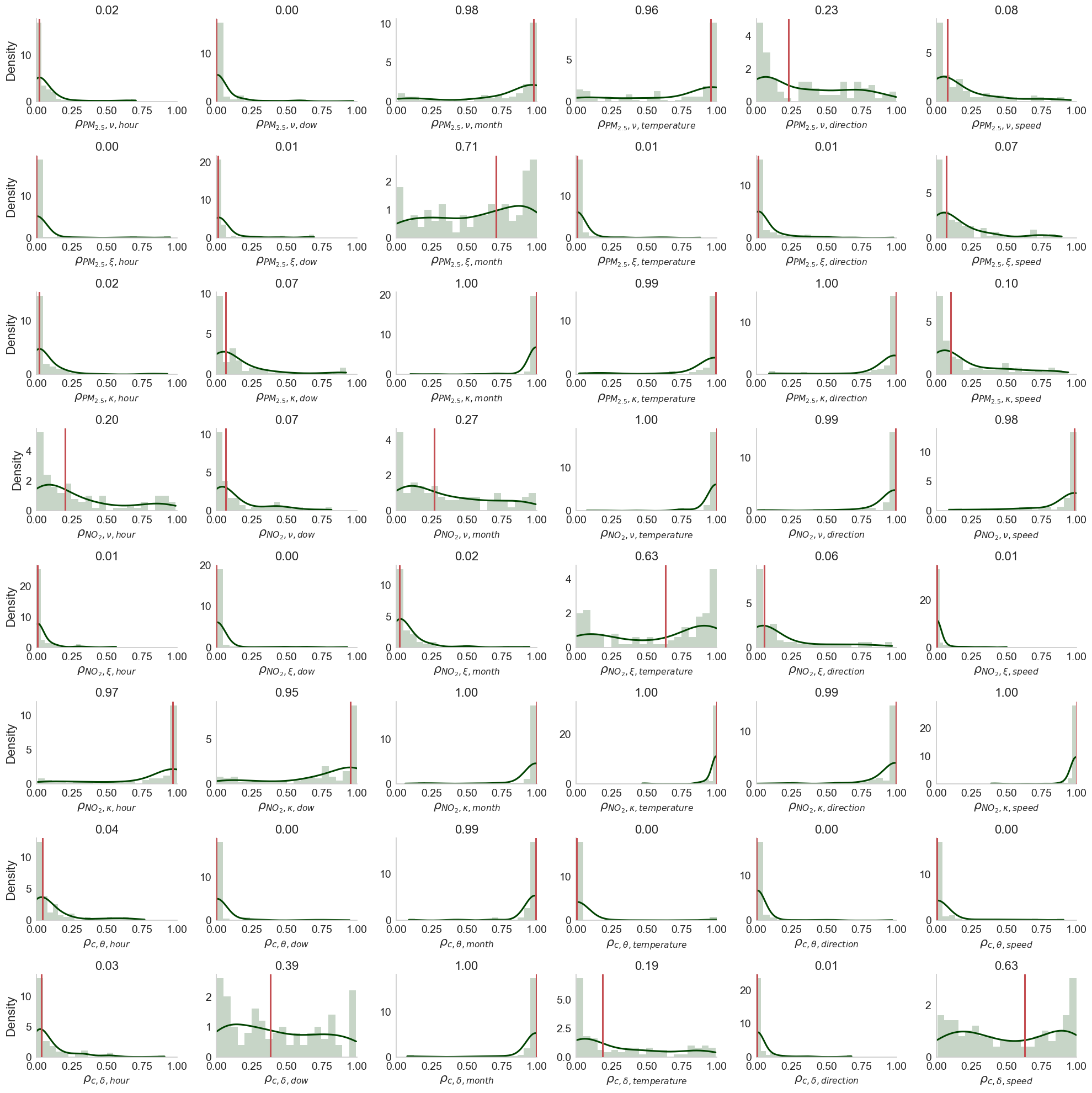}
  \caption{Variational posterior distributions of the inclusion probabilities for all nonlinear effects in the Edinburgh application. Rows: $(\nu,\xi,\kappa)$ of PM$_{2.5}$ (first three), $(\nu,\xi,\kappa)$ of NO$_2$ (next three), and BB1 copula parameters $(\theta,\delta)$ (last two). Columns: hour, day of the week, month, temperature, wind direction, and wind speed. Each panel shows a histogram of sampled inclusion probabilities together with a smooth density estimate; the vertical line marks the posterior median, which is also reported in the panel title.}
  \label{fig:effects_rho_app}
\end{figure}

\subsubsection{Posterior checks}

We assessed the fitted model conditionally on the observed covariates, with the resulting diagnostics reported in Figure~\ref{fig:ppc_application_90}. Marginal fit was evaluated through posterior predictive checks based on replicated responses generated from the fitted EGPD--copula model under draws from the variational approximation. For the joint distribution, we instead evaluated the model-implied conditional probabilities directly from the fitted EGPD marginals and BB1 copula, avoiding additional Monte Carlo variability from response simulation.

\begin{figure}[t]
\centering
\includegraphics[width=\linewidth]{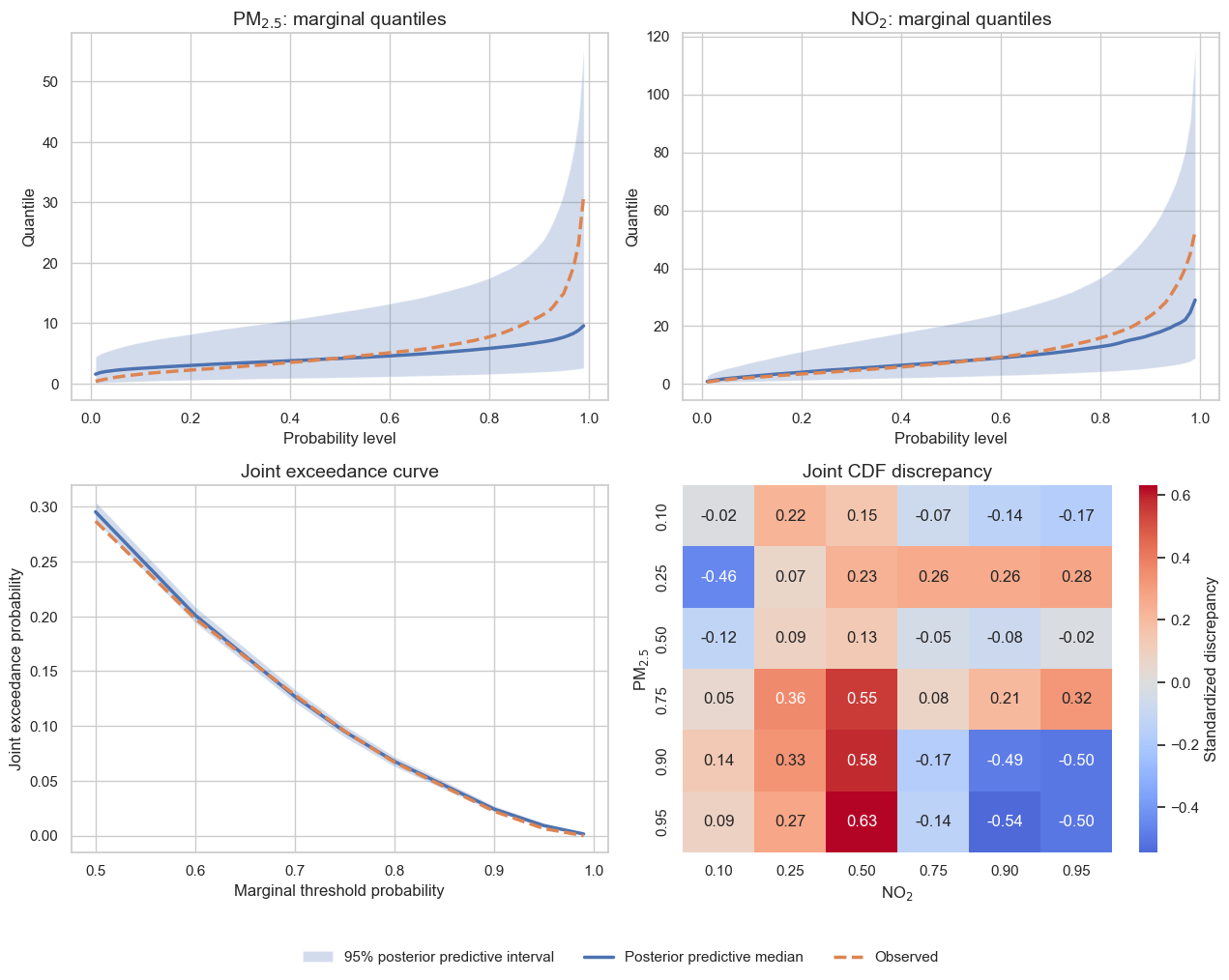}
\caption{Posterior checks for the Edinburgh air-pollution application corresponding to the halo analysis with $\alpha=0.90$, performed conditionally on the observed covariates. Top row: observed marginal quantile curves for PM$_{2.5}$ and NO$_2$ compared with posterior predictive medians and pointwise $95\%$ posterior predictive intervals. Bottom left: observed joint exceedance curve compared with the posterior median and pointwise $95\%$ intervals of the model-implied joint exceedance probabilities, computed from the fitted EGPD marginals and BB1 copula. Bottom right: standardized discrepancy between the observed empirical joint distribution function and the model-implied joint distribution function on a grid of marginal probability levels.}
\label{fig:ppc_application_90}
\end{figure}

The marginal posterior predictive checks indicate that the fitted model reproduces the bulk and central parts of the marginal distributions well. Discrepancies are mainly concentrated in the upper tails, where the observed marginal quantiles tend to lie above the posterior predictive medians, while remaining within the posterior predictive uncertainty bands. The model-implied joint exceedance curve closely tracks the observed curve across marginal threshold levels, suggesting that the fitted copula captures the scale of joint tail probabilities well. The discrepancy grid for the joint distribution function shows moderate localized deviations, but no large systematic lack of fit; the most visible discrepancies occur in the upper part of the grid.

For the joint distribution function check, let $q_d(p)$ denote the empirical $p$-quantile of response $Y_d$, for $d=1,2$. For each pair of probability levels $(p_1,p_2)$, the observed empirical joint distribution function is
\begin{align*}
\widehat F_{\mathrm{obs}}(p_1,p_2)
=
\frac{1}{n}
\sum_{i=1}^n
\mathbbm{1}
\left\{
Y_{i1}\le q_1(p_1),\,
Y_{i2}\le q_2(p_2)
\right\}.
\end{align*}
For posterior draw $s$, the corresponding model-implied joint distribution function is computed as
\begin{align*}
\widehat F_{\mathrm{mod}}^{(s)}(p_1,p_2)
=
\frac{1}{n}
\sum_{i=1}^n
C^{(s)}_{\theta(\mathbf{x}_i),\delta(\mathbf{x}_i)}
\left[
F_1^{(s)}\{q_1(p_1)\mid \mathbf{x}_i\},
F_2^{(s)}\{q_2(p_2)\mid \mathbf{x}_i\}
\right],
\end{align*}
where $F_1^{(s)}$ and $F_2^{(s)}$ are the fitted EGPD marginals and $C^{(s)}$ is the BB1 copula under posterior draw $s$. The heatmap reports the standardized discrepancy
\begin{align*}
D(p_1,p_2)
=
\frac{
\widehat F_{\mathrm{obs}}(p_1,p_2)
-
\operatorname{median}_s\left\{
\widehat F_{\mathrm{mod}}^{(s)}(p_1,p_2)
\right\}
}{
Q_{0.975}\left\{
\widehat F_{\mathrm{mod}}^{(s)}(p_1,p_2)
\right\}
-
Q_{0.025}\left\{
\widehat F_{\mathrm{mod}}^{(s)}(p_1,p_2)
\right\}
}.
\end{align*}
Values close to zero indicate agreement between the observed empirical joint distribution function and the model-implied joint distribution function, relative to posterior uncertainty.

\paragraph{Representative observations within fitted halos}

Tables~\ref{tab:risk_halo_examples_application} and~\ref{tab:favourable_halo_examples_application} report, for each component, the observations with the highest posterior halo-membership probabilities. The examples emphasize the role of temperature: NO$_2$ halo observations are concentrated under colder conditions and predominantly westerly winds, whereas PM$_{2.5}$ examples occur mainly in spring under mild-to-cool temperatures. Joint favourable-halo examples are more heterogeneous, but are often associated with spring or autumn and afternoon or evening hours. In Table~\ref{tab:favourable_halo_examples_application}, the selected observations represent covariate configurations with elevated fitted exceedance probability relative to the background, not necessarily the largest observed pollutant values.

\begin{table}
\centering
\resizebox{\linewidth}{!}{%
\begin{tabular}{@{}cccccccccc@{}}
\toprule
\textbf{Component} & \textbf{Obs.} & \textbf{Hour} & \textbf{Day} & \textbf{Month}
& \textbf{Temp.} & \textbf{Direction} & \textbf{Speed}
& \textbf{PM$_{2.5}$} & \textbf{NO$_2$} \\
& & & & & $(^\circ\mathrm{C})$ & $(^\circ)$ & $(\mathrm{m/s})$
& & \\ \midrule
\multirow{3}{*}{PM$_{2.5}$}
  & 3120 & 11 & 7 & 5  & 10.2 & 160.9 & 1.4 & 8.066  & 2.786  \\
  & 1857 & 6  & 1 & 3  & 6.0  & 230.9 & 2.8 & 12.193 & 15.412 \\
  & 1849 & 21 & 7 & 3  & 6.5  & 80.1  & 2.1 & 22.665 & 14.971 \\ \midrule
\multirow{3}{*}{NO$_2$}
  & 8080 & 23 & 2 & 12 & 0.9  & 277.2 & 5.5 & 7.854  & 52.868 \\
  & 7385 & 13 & 7 & 11 & 4.0  & 260.7 & 5.1 & 2.123  & 15.113 \\
  & 7387 & 15 & 7 & 11 & 4.4  & 236.8 & 3.9 & 2.099  & 12.416 \\ \midrule
\multirow{3}{*}{$c$}
  & 8080 & 23 & 2 & 12 & 0.9  & 277.2 & 5.5 & 7.854  & 52.868 \\
  & 4515 & 9  & 5 & 7  & 14.8 & 107.0 & 3.3 & 1.108  & 2.687  \\
  & 4487 & 4  & 4 & 7  & 10.8 & 241.3 & 1.7 & 5.542  & 13.259 \\ \bottomrule
\end{tabular}%
}
\caption{Representative observations with the highest posterior risk-halo membership probabilities in the Edinburgh application. For each component, the table reports the three observations with the largest posterior halo-membership probabilities. The covariate values characterize the corresponding fitted halo configurations, while the observed PM$_{2.5}$ and NO$_2$ concentrations are reported for reference.}
\label{tab:risk_halo_examples_application}
\end{table}

\begin{table}
\centering
\resizebox{\linewidth}{!}{%
\begin{tabular}{@{}cccccccccc@{}}
\toprule
\textbf{Component} & \textbf{Obs.} & \textbf{Hour} & \textbf{Day} & \textbf{Month}
& \textbf{Temp.} & \textbf{Direction} & \textbf{Speed}
& \textbf{PM$_{2.5}$} & \textbf{NO$_2$} \\
& & & & & $(^\circ\mathrm{C})$ & $(^\circ)$ & $(\mathrm{m/s})$
& & \\ \midrule
\multirow{3}{*}{PM$_{2.5}$}
  & 1421 & 3  & 2 & 3  & 7.4  & 254.0 & 7.9 & 3.821  & 2.227  \\
  & 2025 & 13 & 1 & 3  & 9.2  & 244.2 & 1.6 & 3.679  & 9.682  \\
  & 2019 & 7  & 1 & 3  & 5.8  & 341.2 & 1.1 & 5.590  & 34.433 \\ \midrule
\multirow{3}{*}{NO$_2$}
  & 8080 & 23 & 2 & 12 & 0.9  & 277.2 & 5.5 & 7.854  & 52.868 \\
  & 6530 & 3  & 5 & 10 & 5.5  & 279.5 & 6.4 & 4.575  & 23.034 \\
  & 1313 & 10 & 4 & 2  & 2.5  & 268.7 & 4.3 & 3.656  & 37.459 \\ \midrule
\multirow{3}{*}{$c$}
  & 6430 & 17 & 7 & 10 & 10.2 & 117.4 & 5.2 & 7.429  & 11.609 \\
  & 6133 & 15 & 2 & 9  & 13.2 & 185.4 & 3.9 & 4.882  & 4.074  \\
  & 2513 & 21 & 1 & 4  & 6.5  & 229.5 & 0.5 & 2.642  & 5.486  \\ \bottomrule
\end{tabular}%
}
\caption{Representative observations with the highest posterior risk-halo membership probabilities in the Edinburgh application. For each component, the table reports the three observations with the largest posterior halo-membership probabilities. The covariate values characterize the corresponding fitted halo configurations, while the observed PM$_{2.5}$ and NO$_2$ concentrations are reported for reference.}
\label{tab:favourable_halo_examples_application}
\end{table}

\subsubsection{Covariate dependence within halos}

We first report the empirical correlation matrix of the full covariate sample in Figure~\ref{fig:overall_covariate_correlations}. This provides a baseline for interpreting whether the dependence patterns observed inside the halos are already present in the original covariate distribution or are amplified by conditioning on halo membership.

\begin{figure}[t]
\centering
\includegraphics[width=0.35\linewidth]{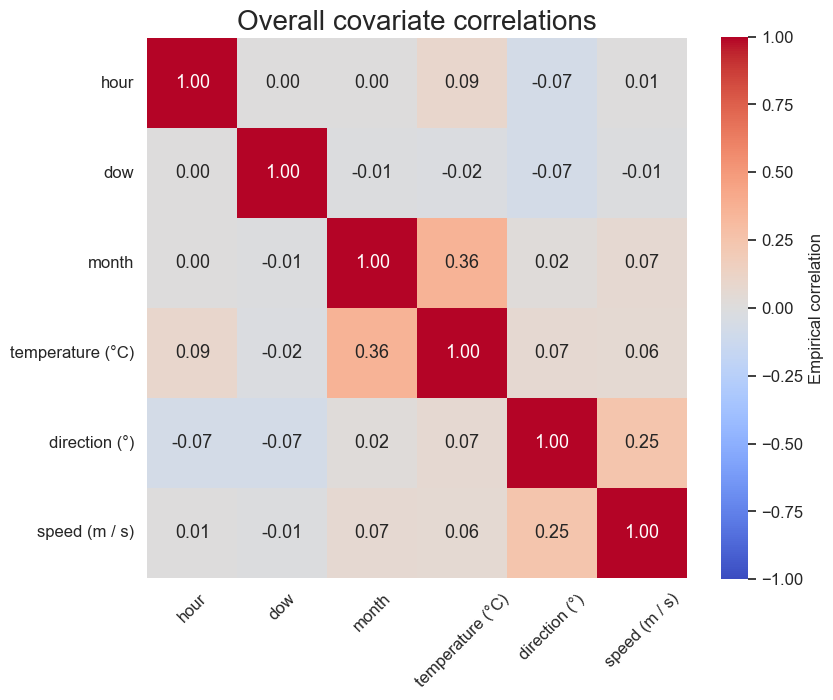}
\caption{Empirical covariate correlations in the complete-case sample.}
\label{fig:overall_covariate_correlations}
\end{figure}

For each posterior draw, we then compute the correlation matrix of the covariates restricted to the observations belonging to the corresponding halo. The entries in Figure~\ref{fig:halo_covariate_correlations_90} report the posterior median of these draw-specific correlations.

\begin{figure}[t]
\centering
\includegraphics[width=\linewidth]{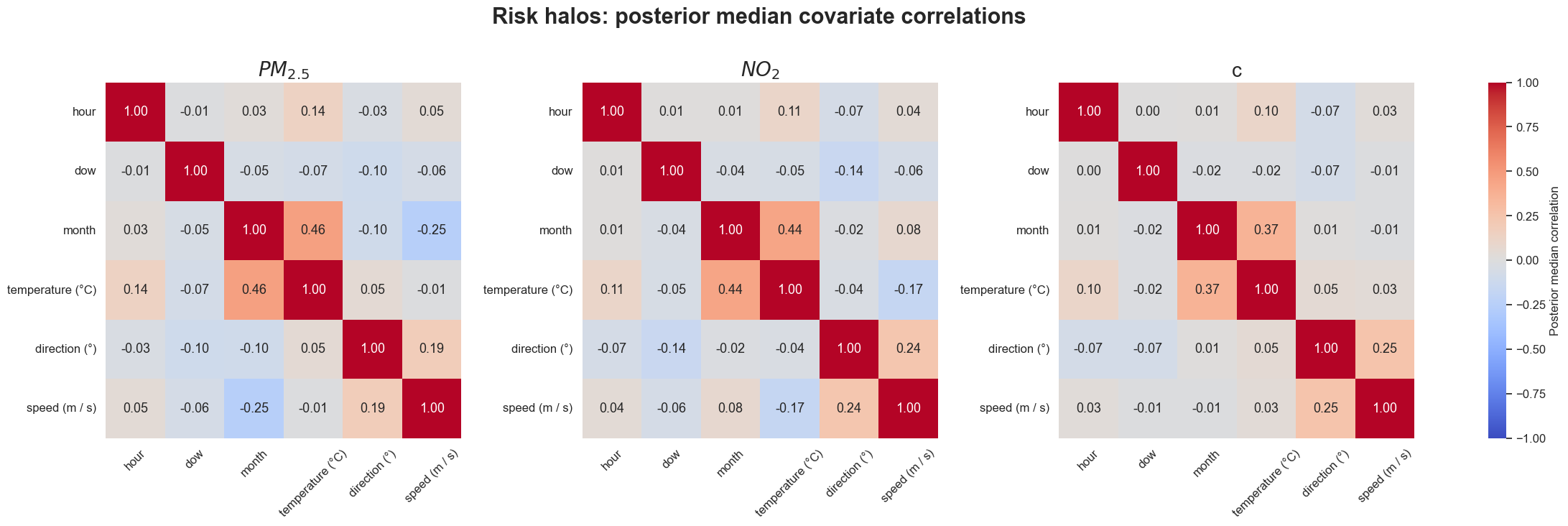}

\includegraphics[width=\linewidth]{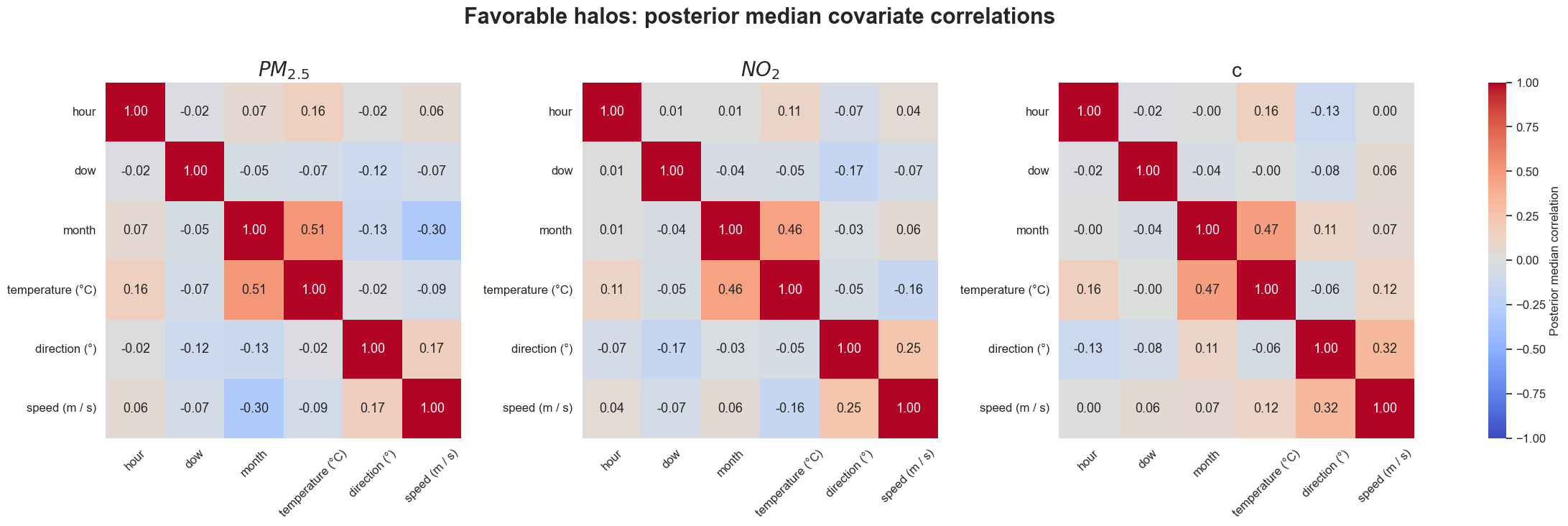}

\caption{Posterior median covariate correlations within the fitted halo representations for the analysis with $\alpha=0.90$. Top: risk halos. Bottom: favorable halos. For each posterior draw, the correlation matrix is computed using the covariate values of the observations belonging to the corresponding draw-specific halo; entries show the posterior median correlation.}
\label{fig:halo_covariate_correlations_90}
\end{figure}

The overall covariate correlations are generally weak, with the clearest background associations being the positive correlation between month and temperature and the positive correlation between wind direction and wind speed. These associations persist inside the fitted halos, but their magnitudes change after conditioning on halo membership. In particular, the month--temperature correlation is stronger inside the halos than in the full sample, especially for the PM$_{2.5}$ favorable halo and the joint favorable halo. This suggests that tail-relevant covariate regions accentuate the seasonal temperature structure present in the background data.

The association between wind direction and wind speed is also visible in the full sample and remains present within the halos. Its magnitude is broadly similar for the NO$_2$ and joint risk halos, while it becomes somewhat stronger in the joint favorable halo. By contrast, for PM$_{2.5}$ halos the most notable change is the negative association between month and wind speed, which is weakly positive in the full sample but becomes negative inside both the risk and favorable halos. This indicates that the PM$_{2.5}$ halo regions are not only characterized by marginal shifts in month and wind speed separately, but also by a different seasonal--wind-speed configuration.

Overall, the correlation patterns remain weak to moderate. This suggests that the fitted halo densities are mainly explained by changes in marginal covariate profiles rather than by strong pairwise dependence among covariates inside the halos. Nevertheless, the comparison with the full-sample correlation matrix shows that conditioning on halo membership can modify specific associations, most clearly for month and temperature, month and wind speed, and wind direction and wind speed.

Since several covariates are cyclic, these Pearson correlations should be interpreted as exploratory summaries rather than as complete measures of circular dependence.

\subsubsection{Sensitivity to the risk-halo level $\alpha$}

\begin{figure}[t]
  \centering
  \includegraphics[width=0.7\linewidth]{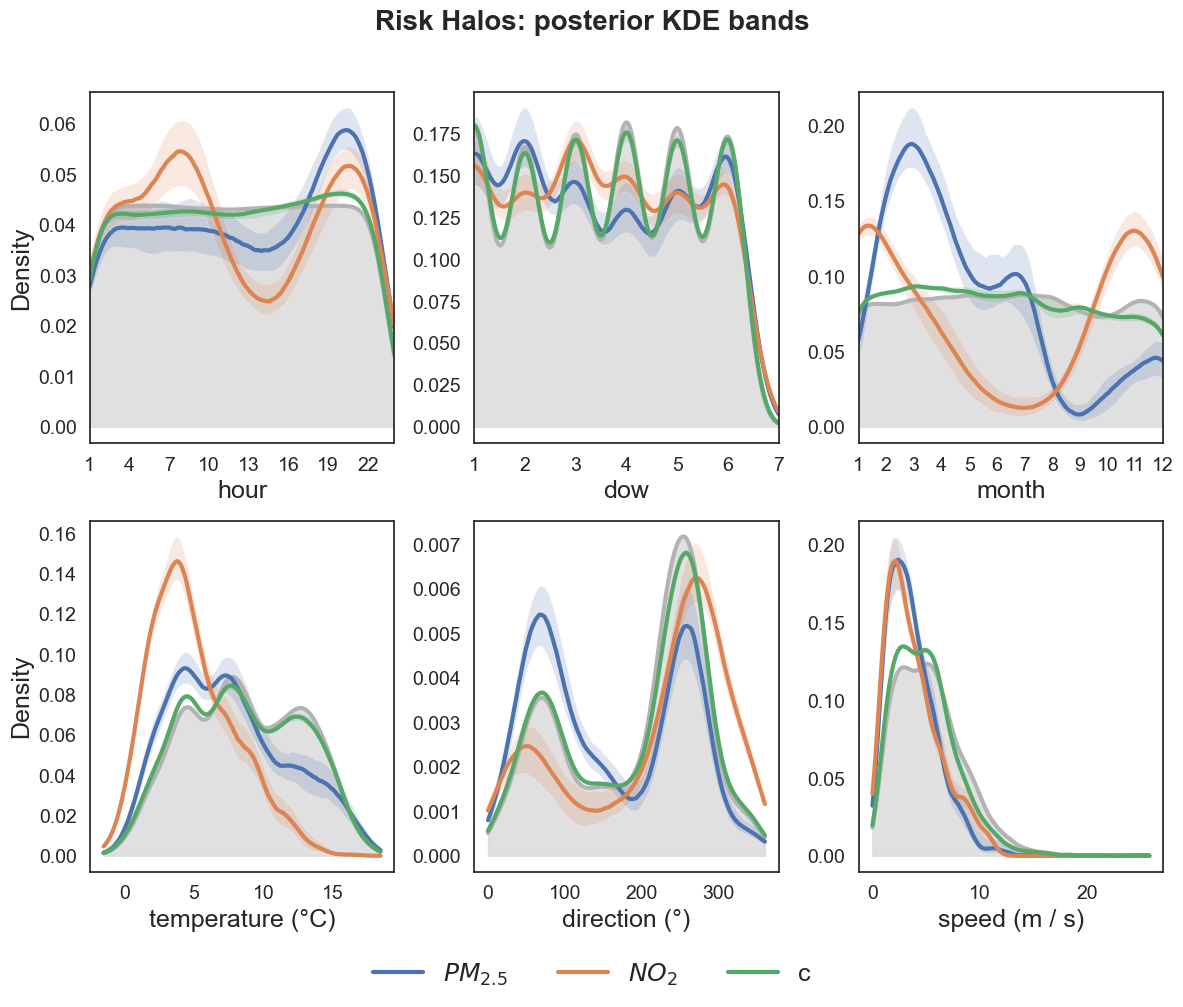}
  \caption{Posterior risk-halo densities for PM$_{2.5}$, NO$_2$, and their joint tail for $\alpha=0.85$ (colored curves with shaded credible bands), compared with the overall covariate distribution (gray).}
  \label{fig:h_risk_app_85}
\end{figure}

\begin{figure}[t]
\centering
\includegraphics[width=\linewidth]{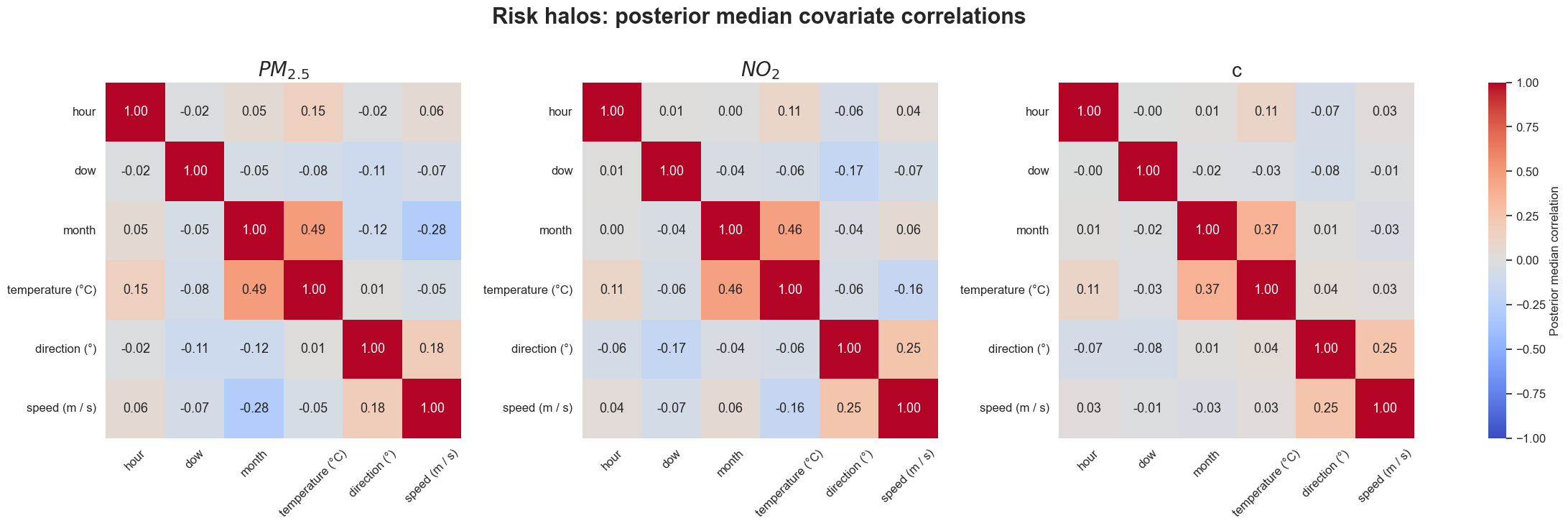}
\caption{Posterior median covariate correlations within the fitted risk halos for $\alpha=0.85$. For each posterior draw, the correlation matrix is computed using the covariate values of the observations belonging to the corresponding draw-specific risk halo; entries show the posterior median correlation.}
\label{fig:risk_halo_covariate_correlations_85}
\end{figure}

\begin{figure}[t]
  \centering
  \includegraphics[width=0.7\linewidth]{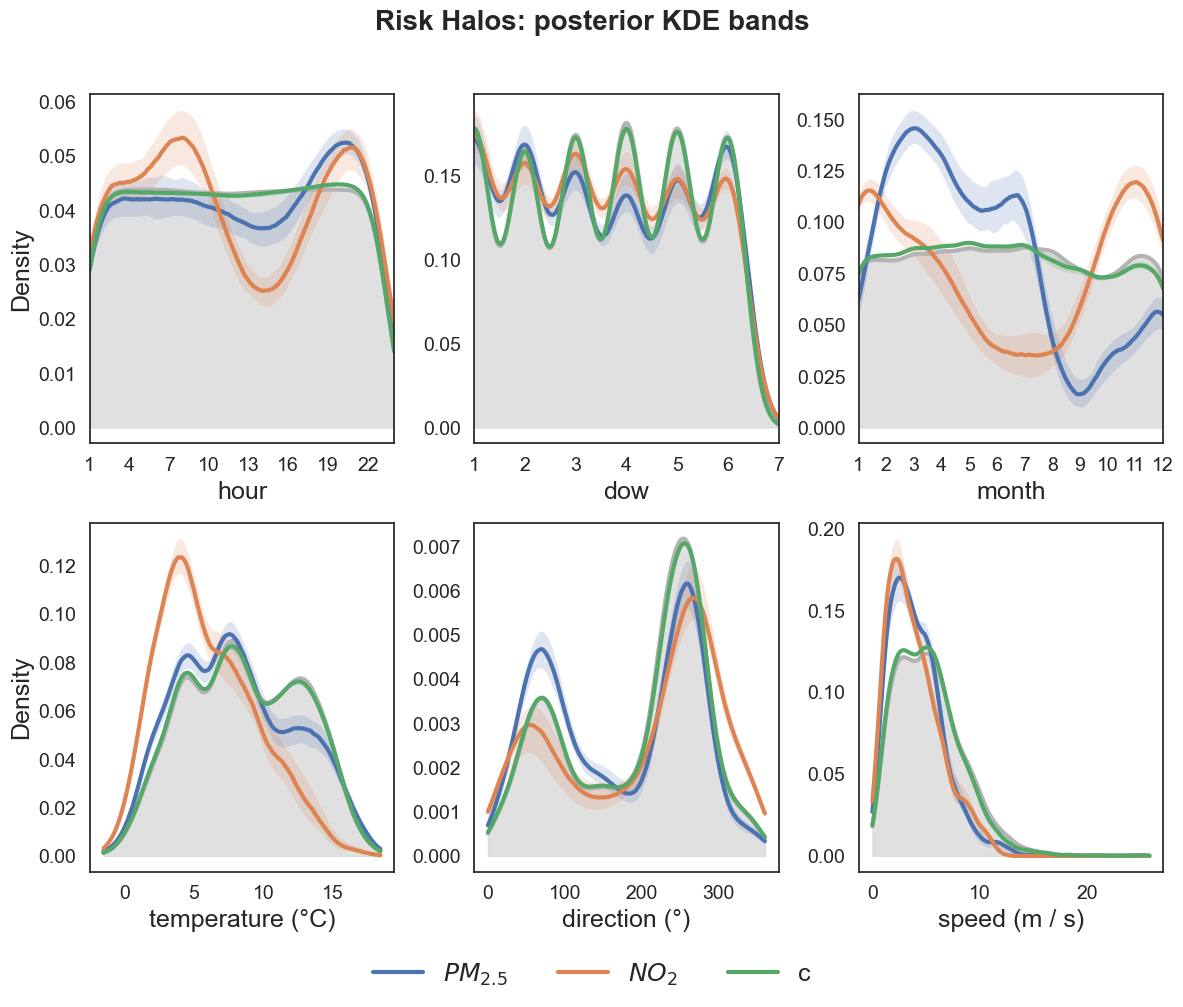}
  \caption{Posterior risk-halo densities for PM$_{2.5}$, NO$_2$, and their joint tail for $\alpha=0.95$ (colored curves with shaded credible bands), compared with the overall covariate distribution (gray).}
  \label{fig:h_risk_app_95}
\end{figure}

\begin{figure}[t]
\centering
\includegraphics[width=\linewidth]{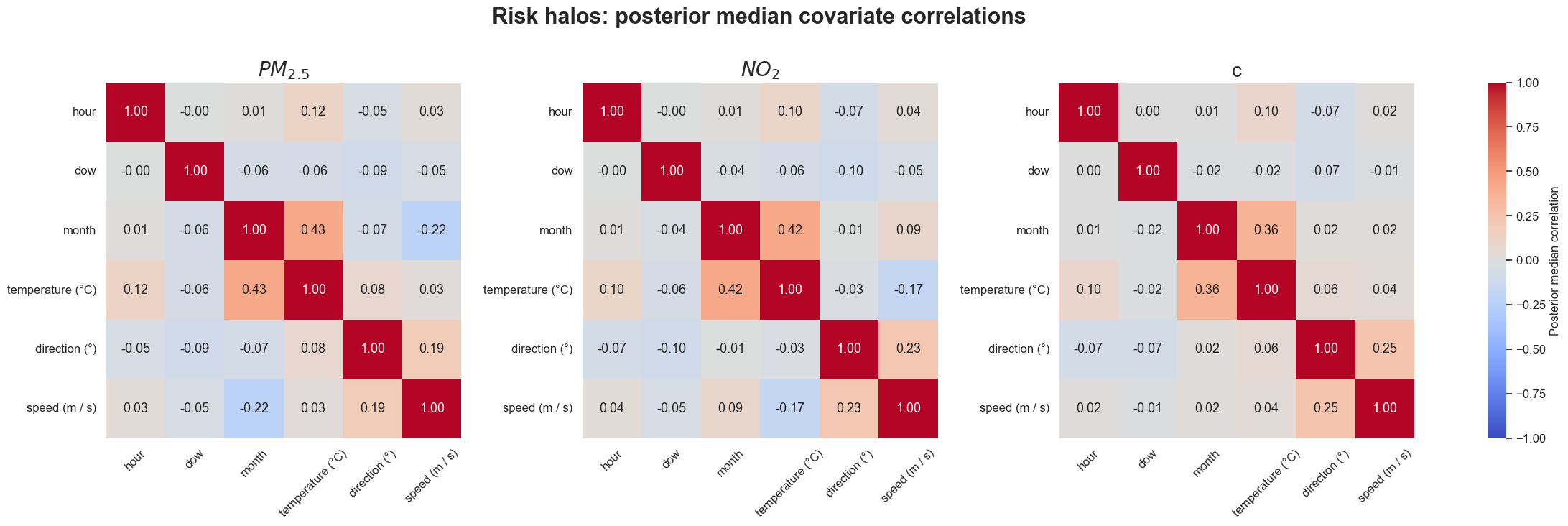}
\caption{Posterior median covariate correlations within the fitted risk halos for $\alpha=0.95$. For each posterior draw, the correlation matrix is computed using the covariate values of the observations belonging to the corresponding draw-specific risk halo; entries show the posterior median correlation.}
\label{fig:risk_halo_covariate_correlations_95}
\end{figure}

\end{document}